\documentclass[journal]{IEEEtran}

\usepackage{xcolor}
\usepackage{amsmath}
\usepackage{amssymb}
\usepackage{amsthm}
\usepackage{graphicx}
\usepackage{booktabs}
\usepackage{multirow}
\usepackage{array}

\usepackage{cite}
\usepackage[hidelinks]{hyperref}
\newcommand{\figref}[1]{Fig.~\ref{#1}}

\theoremstyle{plain}

\newtheorem{proposition}{Proposition}
\newtheorem{lemma}{Lemma}
\newtheorem{corollary}{Corollary}

\theoremstyle{definition}

\title{Structured Edge-Aware Graph Attention Network for Transmitter-Resolved Pointwise Radio Map Estimation}

\author{Ang Li, Chengyu Liu, and Yue Wang%
\thanks{Ang Li, Chengyu Liu, and Yue Wang are with China Telecom Research Institute, China (e-mail: lia15@chinatelecom.cn; liucy32@chinatelecom.cn; yue.wang@chinatelecom.cn).}%
\thanks{Corresponding author: Yue Wang.}}

\begin{document}

\maketitle

\begin{abstract}
Transmitter-resolved radio map estimation (RME) from sparse measurements is essential for obtaining source-specific received-power information in wireless networks. 
This paper proposes SeaGAT, a Structured Edge-Aware Graph Attention Network for transmitter-resolved pointwise RME. 
For each target--transmitter query, SeaGAT constructs a target-centered graph from a bounded transmitter-specific reference set. 
Reference representations generate messages and provide measurement context to attention scoring, whereas structured edge representations encode explicit query--reference relations and enter only the scoring procedure. 
Within a fixed sampled support, we derive a differential identity that decomposes local variation of the evidence aggregate into attention-weighted message changes and score-induced weight redistribution; edge-relation perturbations act only through the latter, whereas a reference-RSS perturbation can affect both pathways.
The bounded star graph yields graph-side computation linear in the selected reference count $K$, while support-truncation analysis bounds deviation from the same-parameter full-support aggregate by the product of omitted attention mass and cross-support message diameter. 
Extensive computer simulations evaluated with ray-tracing dataset show that SeaGAT achieves the lowest mean RMSE compared with baselines; replacing learned attention with uniform weights increases RMSE by $1.58$--$2.14$~dB. It also maintains stable transfer across the tested urban-layout and carrier-frequency shifts.
\end{abstract}

\begin{IEEEkeywords}
Radio map estimation, graph attention networks, sparse measurements, received signal strength, transmitter-resolved radio maps, relation-conditioned evidence aggregation.
\end{IEEEkeywords}

\section{Introduction}\label{sec:intro}

\IEEEPARstart{R}{adio} maps (RMs) describe the spatial distribution of radio quantities such as received signal strength (RSS), power spectral density, or channel gain over a geographic region \cite{2023_SpectrumSurvying,2022_Magzine,2024_TheoreticalAnalysisRME, 2024_ChannelKnowledgeMap}. 
They provide a compact representation of the radio environment and support a variety of operations, including network planning \cite{2022_NetworkPlanning}, resource management \cite{2019_ResourceManagement}, localization \cite{2023_Localization}, channel state information feedback reduction \cite{2025_CSIFeedbackReduction}, V2X mobility services \cite{2024_LSTM}, and UAV-assisted communication tasks such as trajectory planning \cite{2023_UAV}. 
In practical scenarios, dense radio fields are rarely available directly; instead, RME aims to infer unobserved radio quantities from sparse and irregular measurements. 
This scarcity has also motivated sample-efficient acquisition strategies, e.g., Bayesian active learning for selecting informative measurements \cite{2024_BayesianActiveLearning}.

Most existing RME formulations focus on reconstructing a dense spatial field. However, many practical applications require the received power associated with a designated transmitter at selected target locations, rather than the complete map over the entire area, e.g., when evaluating candidate serving cells, querying link quality, or comparing transmitter-specific coverage at specific points. 
This distinction is important in multi-transmitter environments where an occupancy radio map represents the aggregated contribution of multiple transmitters, whereas a transmitter-resolved radio map describes the propagation field associated with an individual transmitter \cite{2024_TheoreticalAnalysisRME,2022_Magzine}. 
From an algorithmic perspective, the difference is reflected by whether transmitter-related characteristics are explicitly incorporated into the inference process. 
Many RME methods treat the radio map as a generic spatial field and can be applied to either aggregated or transmitter-specific maps once measurements are provided. 
However, transmitter-resolved estimation requires modeling source-dependent relevance because the spatial correlation of an individual transmitter's field can be jointly affected by transmitter geometry, propagation paths, and blockage conditions.
A transmitter-agnostic formulation may reconstruct spatial variations effectively, but it does not explicitly describe how reference evidence should be transferred for a designated transmitter query. 
This paper focuses on transmitter-resolved pointwise RSS estimation and investigates how transmitter-specific information can be integrated into sparse measurement inference.

\begin{figure}[!t]
	\centering
	\includegraphics[width=0.45\textwidth]{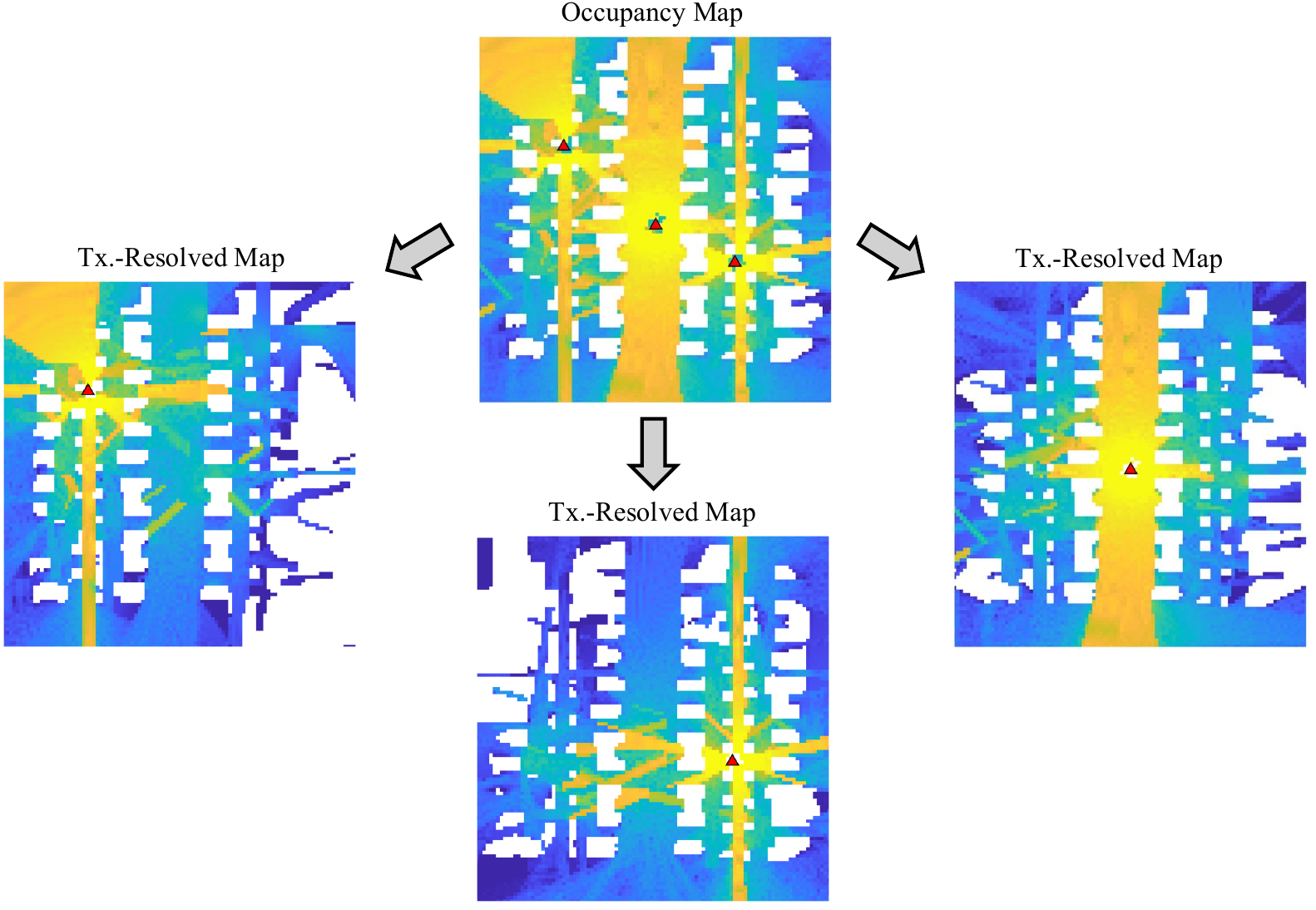}
	\caption{Illustration of the difference between the occupancy radio map and transmitter-resolved maps. A transmitter-resolved radio map describes the spatial received power associated with one designated transmitter (red triangles), whereas an occupancy radio map summarizes the aggregated radio power over the same area.}
	\label{fig:rm_concept}
\end{figure}

Classical RME methods typically estimate unobserved locations by imposing structural priors on the spatial field. 
Depending on the adopted prior, classical RME methods infer missing values by exploiting different forms of spatial structures. Representative examples include kernel-based interpolation methods, such as inverse distance weighting (IDW) \cite{2008_InverDistanceWeight}, which utilize predefined spatial similarity functions; kriging methods model spatial correlations through covariance functions \cite{1990_Kriging,2003_KrigingBook}; and matrix or tensor completion approaches recover discretized radio maps based on low-rank assumptions \cite{2022_MatrixCompletion,2022_TensorCompletion,2017_SVM,2017_RBF}.
Integrated formulations have further coupled spatial interpolation with matrix completion in a unified convex-optimization model \cite{2025_MatrixCompletion}. Hybrid approaches combine statistical interpolation with data-adaptive components, including neural-network residual kriging \cite{2019_KrigingFFNN} and KNNGP, which combines local neighbor selection with Gaussian-process modeling for radio map construction \cite{2026_KNNGaussian_RME}. 
These approaches' inference mechanism remains governed by manually selected structural priors and does not explicitly model transmitter-dependent propagation characteristics.

Learning-based RME methods replace handcrafted spatial priors with trainable representations learned from measurement data. 
Early approaches retain grid-based output representations and use convolutional reconstruction networks, encoder--decoder architectures, completion networks, or generative models to recover dense radio maps \cite{2018_RecNet,2021_RadioUNet,2022_DCAE,2023_RME_GAN,2024_RobUNet,2026_DeformAttention_RME,2025_2Stage,2023_CellLevel}. 
Some methods additionally incorporate environmental information or transmitter locations as auxiliary inputs. 
For instance, RadioUNet explicitly conditions pathloss prediction on transmitter locations and urban environments \cite{2021_RadioUNet}; DAT-Unet uses deformable attention to adapt feature extraction to sparse spatial observations while retaining a U-Net-style map output \cite{2026_DeformAttention_RME};
and physics-inspired generative frameworks introduce propagation-related features to guide map reconstruction \cite{2025_2Stage}. 
These methods improve nonlinear representation capability; however, their primary objective remains dense map generation rather than query-specific transmitter-resolved inference. 
Transfer learning has also been studied for adapting radio-map predictors across antenna-tilt configurations \cite{2020_TransferLearning}.

Recent studies further relax the fixed-grid assumption by directly predicting radio quantities at arbitrary locations. 
ICR learns continuous representations of emitter propagation patterns \cite{2025_GridFreeUnsupervised_RME}, while attention-based query estimators such as STORM \cite{2025_STORM} and CGFormer \cite{CGFormer} infer target values from sparse reference measurements without predefined grids. 
These methods are aligned with pointwise prediction tasks, but the aggregation between reference measurements and target queries is mainly determined by generic spatial representations. As a result, query-specific relations between reference observations and designated target links remain implicit.

On the other hand, graph-based learning provides a natural framework for representing sparse and irregular measurements \cite{2018_GraphSignalProcess}. Existing graph methods exploit geometric, environmental, or propagation-related correlations to improve radio map reconstruction. For instance, 5GNN constructs local and global measurement graphs to capture spatial dependencies \cite{2022_5GNN}; RadioGAT incorporates propagation-model-based graph encoding for multi-band radiomap reconstruction \cite{2024_RadioGAT}; and Dual-GRE combines physics-enhanced graph attention with CNN refinement for dense reconstruction \cite{2025_DualGRE}. These studies demonstrate the effectiveness of graph learning for RME. Nevertheless, existing graph formulations are mainly designed for dense reconstruction, multi-band completion, or smoothing over pre-constructed graphs, rather than explicitly learning transmitter-conditioned relations for pointwise queries.
For transmitter-resolved pointwise RME, the central question is how to determine which reference observations are relevant to a target location under a designated transmitter. 
This motivates an explicit architectural decomposition between the message and relation pathways. 
Node representations carry measurement-bearing evidence and also provide observation context to the relevance score, while edge representations encode physical query--reference relations that modulate the score without entering the transferred message. 
Existing attention and graph estimators generally leave this transmitter-conditioned relation pathway implicit.

To address this gap, we propose a Structured Edge-Aware Graph Attention Network (SeaGAT) for transmitter-resolved pointwise RSS estimation.
For each target--transmitter query, SeaGAT samples a bounded transmitter-specific reference set based on which a target-centered graph is constructed. 
Reference observations are encoded through node representations, while structured edge representations capture query--reference relations. 
SeaGAT uses compact geometric relations available from the query, references, and transmitter. 
The attention mechanism combines node and edge information to produce relation- and measurement-dependent weights, while only node-derived messages are transferred to the target. 
This separation makes the aggregation mechanism analytically tractable. We interpret its softmax weights as an entropy-regularized barycentric allocation and derive a unified differential identity that distinguishes direct message variation from score-induced weight redistribution. 
It is shown that edge-relation changes shift the aggregate through the weighting pathway, whereas an RSS perturbation acts through both message content and relevance scoring. 
We further quantify how bounded reference support deviates from the same-parameter full-support operator.

The main contributions of this paper are summarized as follows:
\begin{itemize}
    \item We formulate transmitter-resolved pointwise RME as a query-conditioned evidence transfer problem and propose a bounded target-centered star graph accordingly. Transmitter-specific observations generate node messages, while explicit physical relations act through the attention score to control how those messages are transferred to the queried target.

    \item Using an entropy-regularized barycentric view of softmax, we establish how SeaGAT's explicit message--relation separation governs reference influence. Edge relations alter the aggregate solely through normalized weight redistribution, whereas a reference RSS affects both the transferred message and its relevance score; we also analyze the effect of bounded support by the omitted full-support attention mass and the separation between retained and omitted messages.

    \item We develop compact relation encodings that combine compressive distance bucketization, circular bearing interpolation, and optional link-state agreement for efficient query-dependent inference.
\end{itemize}

\section{Preliminaries and Problem Setup}\label{sec:system}

This section first formalizes the transmitter-resolved pointwise RME problem and identifies the design requirements it imposes on the estimator. 
It then introduces the target-centered graph sampling and the transmitter-conditioned node and edge attributes used later in the architecture.

\subsection{Transmitter-Resolved Pointwise RME}
\label{subsec:problem_formulation}
Consider a geographical area $A \subset \mathbb{R}^2$ where a set of transmitters $B$ are deployed.
For each transmitter $b \in B$, let $\mathbf{p}_b \in \mathbb{R}^2$ denote its coordinates.
Each query target $t$ is associated with a location $\mathbf{p}_t \in A$, and its received signal power with respect to transmitter $b$ is denoted by $y_t^{(b)}$.
The objective is to estimate received power only for selected target--transmitter queries directly.
Let $r$ denote the reference observation index and let $R_b^{\mathrm{obs}}$ denote the index set of transmitter-specific reference observations for transmitter $b$.
The corresponding sparse observation set is expressed by
\begin{equation}
O_b
=
\left\{
\left(\mathbf{p}_r, y_r^{(b)}\right) : r \in R_b^{\mathrm{obs}}
\right\}.
\end{equation}
The transmitter-resolved pointwise estimation problem for an arbitrary target location\footnote{We consider only outdoor target locations in this paper.}
is written as
\begin{equation}
\hat{y}_t^{(b)} = f_\Theta\!\left(\mathbf{p}_t,\mathbf{p}_b; O_b\right),
\end{equation}
where $f_\Theta(\cdot)$ denotes the estimator parameterized by $\Theta$.

In our considered setting, the estimator design is therefore centered on how to extract query-relevant evidence from the transmitter-specific observation set. 
This leads to three requirements for the estimator architecture.
First, the estimator should remain transmitter-aware, since the same target location may correspond to very different received power values for different designated transmitters in a multi-transmitter environment.
Second, once a query-dependent local reference set is formed, the selected references have no canonical ordering.
Hence, the estimator should not change its output merely because the same references are listed in a different order \cite{2017_DeepSets}, and it should also remain meaningful when the number of participating references varies.
Third, given a target location and a designated transmitter, the estimator still needs to determine which transmitter-specific references should contribute and how strongly they should be weighted.
In this process, the observed RSS at a reference is the evidence to be transferred, whereas quantities such as the target--reference distance, the bearing, and, when available, the consistency between the propagation patterns experienced by the target and reference links serve as relation attributes that indicate the relevance of that reference to the current query.

\subsection{Target-Centered Graph Sampling}

SeaGAT operates on a query-dependent reference set generated by a bounded sampling operator,
\begin{equation}
\label{ReferenceSet}
N_t^{(b)}
=
\mathcal{S}_K\!\left(\mathbf{p}_t,O_b\right),
\qquad
\left|N_t^{(b)}\right|=K,
\end{equation}
where $\mathcal{S}_K$ selects $K$ transmitter-specific references for the current query.
The downstream encoders and attention module use the attributes of the selected references and therefore support different implementations of $\mathcal{S}_K$.
This work instantiates the operator by Euclidean $k$-nearest-neighbor sampling,
\begin{equation}
\mathcal{S}_K\!\left(\mathbf{p}_t,O_b\right)
=
\mathrm{kNN}_K\!\left(\mathbf{p}_t,O_b\right),
\end{equation}
which controls the graph size and favors references geographically close to the target.
The observation pool $O_b$ can be formed through different acquisition policies, including random sampling and cluster-based representative sampling; their effects are evaluated in Section~\ref{sec:results}.

Based on $N_t^{(b)}$, the target-centered graph is defined as
\begin{equation}
\mathcal{G}_t^{(b)} = \left(\mathcal{V}_t^{(b)}, \mathcal{E}_t^{(b)}\right),
\end{equation}
with
\begin{align}
\mathcal{V}_t^{(b)} &= \{t\} \cup N_t^{(b)},\\
\mathcal{E}_t^{(b)} &= \left\{(r,t) : r \in N_t^{(b)}\right\}.
\end{align}
The directed target-centered star topology follows from the pointwise task i.e., every edge represents reference-to-target evidence transfer for one queried target, while the sampling operator determines the candidate evidence support. The current $k$-NN implementation adds a geographical locality prior and bounds the per-query computation. 
Other bounded query-dependent samplers can replace $\mathcal{S}_K$ without changing the subsequent encoders or aggregation. 
Accordingly, the graph contains only reference-to-target edges, and each query is answered by direct aggregation over its bounded evidence support.

\subsection{Node and Edge Attributes}
Since each query is conditioned on a designated transmitter, SeaGAT expresses node and edge attributes in a transmitter-conditioned form.
For each $r \in N_t^{(b)}$, we calculate its transmitter-centered coordinate as $\mathbf{p}_r^{(b)} = \mathbf{p}_r - \mathbf{p}_b$, and the corresponding reference observation vector is expressed as
\begin{equation}
\label{NodeFeature}
\mathbf{u}_r^{(b)}
=
\left[
\left(\mathbf{p}_r^{(b)}\right)^\top,\;
y_r^{(b)}
\right]^\top.
\end{equation}

For each sampled directed edge $(r,t)$, we write its target--reference distance and bearing relations as
$d_{t,r} = \left\|\mathbf{p}_r - \mathbf{p}_t\right\|$
and $\theta_{t,r} = \mathrm{wrap}\!\big(\angle(\mathbf{p}_r - \mathbf{p}_t)\big)$,
where $\mathrm{wrap}(\cdot)$ maps the angle to $(-\pi,\pi]$.
The corresponding geometric relation vector is then written as
\begin{equation}
\label{EdgeFeature}
\boldsymbol{\phi}_{t,r}
=
\left[
d_{t,r},\;
\theta_{t,r}
\right]^\top.
\end{equation}
While \eqref{EdgeFeature} adopts the target--reference distance and bearing as the default relation set, transmitter-anchored relations of the same form---e.g., the reference--transmitter and target--transmitter distance--bearing pairs---can be appended without modifying the encoders introduced in Section~\ref{sec:method}; the effect of this extended relation set is evaluated in Section~\ref{sec:results}.

When link-state indicators are available as side information, the propagation-pattern consistency identified in Section~\ref{subsec:problem_formulation} can be instantiated as a link-state agreement variable.
Let $\ell_t^{(b)},\, \ell_r^{(b)} \in \{\mathrm{LoS}, \mathrm{NLoS}\}$
denote the link state of the target--transmitter and reference--transmitter links, respectively.
Their agreement pattern is defined as
\begin{equation}
\label{StateAgreement}
a_{t,r}^{(b)}
=
\begin{cases}
\mathrm{LL}, & \ell_t^{(b)}=\mathrm{LoS},\;\; \ell_r^{(b)}=\mathrm{LoS},\\
\mathrm{NN}, & \ell_t^{(b)}=\mathrm{NLoS},\; \ell_r^{(b)}=\mathrm{NLoS},\\
\mathrm{LN}, & \ell_t^{(b)}=\mathrm{LoS},\;\; \ell_r^{(b)}=\mathrm{NLoS},\\
\mathrm{NL}, & \ell_t^{(b)}=\mathrm{NLoS},\; \ell_r^{(b)}=\mathrm{LoS}.
\end{cases}
\end{equation}
It is worth noting that this variable differs from the inter-node link connectivity used in building-aware graph construction \cite{2023_GNNLetter}.
When such link-state indicators are unavailable, SeaGAT operates using geometric relation attributes only.

The variables in \eqref{NodeFeature}--\eqref{StateAgreement} enter message passing through complementary representations.
The reference vector $\mathbf{u}_r^{(b)}$ carries measurement-bearing evidence, while $\boldsymbol{\phi}_{t,r}$ and $a_{t,r}^{(b)}$ describe the reference relation to the current target--transmitter query.
The node encoder maps $\mathbf{u}_r^{(b)}$ to $\mathbf{h}_r^{(b)}$, and the edge encoder maps the relation variables to $\mathbf{e}_{t,r}^{(b)}$ in Section~\ref{sec:method}.
The resulting processing pipeline is summarized in \figref{fig:blockdiagram}.

\section{Structured Edge-Aware GAT Architecture}\label{sec:method}

\begin{figure}[!t]
	\centering
	\includegraphics[width=0.47\textwidth]{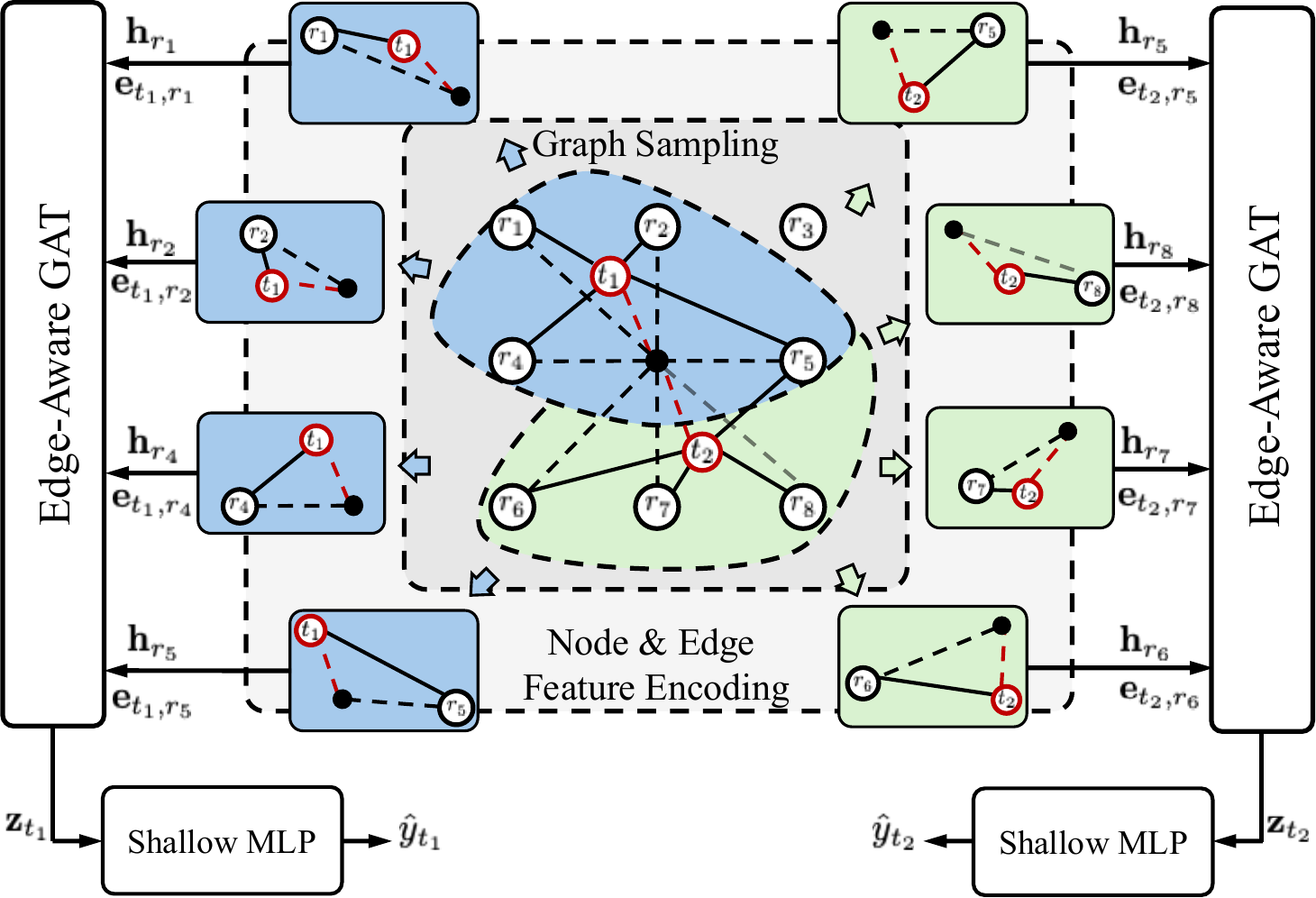}
	\caption{Overview of SeaGAT. For each target--transmitter query, SeaGAT performs graph sampling from a transmitter-specific reference pool, derives node features and target--reference edge features, applies relation-aware graph attention to aggregate transmitter-specific evidence, and outputs the pointwise RSS estimate through a shallow MLP. The transmitter index $b$ is omitted for simplicity, and the black-filled circle denotes the transmitter. }
\label{fig:blockdiagram}
\end{figure}

\figref{fig:blockdiagram} illustrates the proposed SeaGAT architecture through a toy example. 
Specifically, two targets are queried against a transmitter-specific reference pool of size $|O_b|=8$.
For each queried target, $K=4$ nearby references are selected by the $k$-NN sampling to form the target-centered local graph.
The reference observations and the target--reference relations are then encoded into node and edge representations, respectively, and fused by the structured edge-aware attention to aggregate the transmitter-specific evidence. 
Finally, a shallow MLP maps the aggregated target representation to the pointwise RSS estimate.

\subsection{Node and Edge Feature Encoding}

To facilitate the following discussion, we define $\mathrm{Lin}_{\psi}(\mathbf{x}) \triangleq \mathbf{W}_{\psi}\mathbf{x} + \mathbf{b}_{\psi}$ as a linear transformation indexed by $\psi$, let $\mathrm{MLP}_{\psi}(\cdot)$ denote a lightweight multilayer perceptron, and let $\mathbf{o}_k$ denote the one-hot vector associated with bucket index $k$.
Recall from \eqref{NodeFeature} that each reference-side observation vector $\mathbf{u}_r^{(b)}$ consists of the transmitter-centered coordinate $\mathbf{p}_r^{(b)}$ and the observed RSS $y_r^{(b)}$.
SeaGAT encodes these two components through separate branches before fusion:
\begin{align}
\mathbf{g}_r^{(b)} &= \mathrm{Lin}_{u}\!\left(\mathbf{p}_r^{(b)}\right),\\
\mathbf{q}_r^{(b)} &= \mathrm{Lin}_{y}\!\left(y_r^{(b)}\right).
\end{align}
The final node feature vector takes the following form
\begin{equation}
\mathbf{h}_r^{(b)}
=
\mathrm{MLP}_{\mathrm{ref}}
\!\left(
\left[
\left(\mathbf{g}_r^{(b)}\right)^\top,\;
\left(\mathbf{q}_r^{(b)}\right)^\top
\right]^\top
\right).
\end{equation}

For each sampled directed edge $(r,t)$, SeaGAT maps the edge-side relation variables to a compact representation through quantized encodings.
The distance relation is encoded by uniform bucketization after a monotone companding transform, i.e.,
\begin{equation}
\label{DistEmbed}
\mathbf{v}_{t,r}^{(d)}
=\mathbf{W}_{d}\,\mathbf{o}_{Q_d(d_{t,r})},
\end{equation}
where $Q_d(d_{t,r})$ denotes the resulting bucket index, i.e.,
\begin{equation}
\label{DistQuant}
Q_d(d_{t,r})
=
\left\lfloor
B_d\,\frac{d_{t,r}}{c+d_{t,r}}\right\rfloor+1.
\end{equation}
$B_d$ denotes the configured number of distance buckets, and $c>0$ is a fixed range-scale constant whose value is specified in Section~\ref{sec:results}.
The bearing relation is periodic and is therefore encoded by a circular soft quantization.
Specifically, the bearing is first mapped to a fractional bucket position
\begin{equation}
\label{AngPos}
\rho_{t,r}
=
\frac{B_{\theta}\left(\theta_{t,r}+\pi\right)}{2\pi}
\bmod B_{\theta},
\end{equation}
where $B_{\theta}$ denotes the number of bearing buckets.
The encoded bearing relation is then obtained by linearly interpolating the embeddings of the two adjacent circular buckets, i.e.,
\begin{equation}
\label{AngEmbed}
\mathbf{v}_{t,r}^{(\theta)}
=
\left(1-\lambda_{t,r}\right)\mathbf{W}_{\theta}\,\mathbf{o}_{k_{t,r}^{-}}
+
\lambda_{t,r}\,\mathbf{W}_{\theta}\,\mathbf{o}_{k_{t,r}^{+}},
\end{equation}
where $k_{t,r}^{-} = \lfloor \rho_{t,r} \rfloor + 1$, $k_{t,r}^{+} = \left(k_{t,r}^{-} \bmod B_{\theta}\right) + 1$, and $\lambda_{t,r} = \rho_{t,r} - \lfloor \rho_{t,r} \rfloor$.

The two quantizers address different relation geometries. The distance mapping in \eqref{DistQuant} is a uniform quantizer after the monotone warping
\begin{equation}
u(d)=\frac{d}{c+d},\qquad d\geq 0,
\end{equation}
which maps the unbounded distance axis to $[0,1)$; and its inverse-warped boundary takes the following form,
\begin{equation}
\tau_j=\frac{cj}{B_d-j},
\end{equation}
with $j=0,\ldots,B_d-1$. 
For $k=1,\ldots,B_d-1$, the width of bucket $k$ can be obtained as
\begin{equation}
\Delta d_k=\tau_k-\tau_{k-1}
=\frac{cB_d}{(B_d-k)(B_d-k+1)}.
\label{eq:bucket_width}
\end{equation}
The equation above indicates that $\Delta d_k$ increases w.r.t $k$, such that the first bucket has width $c/(B_d-1)\approx c/B_d$, and the final bucket extends to infinity. The encoding therefore requires no scenario-dependent clipping and allocates the finest resolution to nearby references. 
For bearing, the interpolation in \eqref{AngEmbed} is continuous and periodic across both bucket boundaries and the $\pm\pi$ wrap-around while retaining table-lookup efficiency; adjacent-bin soft assignment is a standard way to suppress quantization artifacts for periodic variables \cite{2005_HOG}.
As for the link-state agreement, the categorical variable $a_{t,r}^{(b)}$ defined in \eqref{StateAgreement} is mapped to a learned embedding, i.e.,
\begin{equation}
\label{LinkEmbed}
\mathbf{v}_{t,r}^{(\ell)}
=
\mathbf{w}_{a_{t,r}^{(b)}},
\end{equation}
where one embedding vector is learned for each agreement category in $\{\mathrm{LL},\mathrm{NN},\mathrm{LN},\mathrm{NL}\}$.
Accordingly, the edge feature vector is written as
\begin{equation}
\mathbf{r}_{t,r}^{(b)}
=
\begin{cases}
\left[
\mathbf{v}_{t,r}^{(d)\top},\;
\mathbf{v}_{t,r}^{(\theta)\top}
\right]^\top, & \text{w/o link state},
\\[2mm]
\left[
\mathbf{v}_{t,r}^{(d)\top},\;
\mathbf{v}_{t,r}^{(\theta)\top},\;
\mathbf{v}_{t,r}^{(\ell)\top}
\right]^\top, & \text{otherwise}.
\end{cases}
\end{equation}
The final edge representation is then
\begin{equation}
\mathbf{e}_{t,r}^{(b)}
=
\mathrm{MLP}_{\mathrm{edge}}\!\left(\mathbf{r}_{t,r}^{(b)}\right).
\end{equation}

\subsection{Relation-Aware Message Passing}
Recall that node representations determine the transferred messages and provide measurement context to the attention score, whereas edge representations enter only the scoring pathway. Physical relations therefore redistribute the barycentric coefficients assigned to existing node messages without directly injecting message content.

For each $r \in N_t^{(b)}$, the attention score is computed as
\begin{equation}
\label{eq:attention_score}
s_{t,r}^{(b)}
=
\mathbf{w}_{\mathrm{att}}^\top
\psi
\!\left(
\mathrm{Lin}_{\mathrm{att}}
\!\left(
\left[
\left(\mathbf{h}_r^{(b)}\right)^\top,\;
\left(\mathbf{e}_{t,r}^{(b)}\right)^\top
\right]^\top
\right)
\right),
\end{equation}
where $\psi(\cdot)$ denotes the LeakyReLU activation.
The normalized attention weight is then
\begin{equation}
\label{eq:attention_weight}
\alpha_{t,r}^{(b)}
=
\frac{
\exp\!\left(s_{t,r}^{(b)}/\tau\right)
}{
\sum_{r' \in N_t^{(b)}}
\exp\!\left(s_{t,r'}^{(b)}/\tau\right)
},
\end{equation}
with $\tau>0$ denoting a fixed temperature hyperparameter.
The node-only message is
\begin{equation}
\mathbf{m}_r^{(b)} = \mathrm{Lin}_{\mathrm{msg}}\!\left(\mathbf{h}_r^{(b)}\right).
\end{equation}
Define the pre-activation evidence aggregate as
\begin{equation}
\label{eq:preactivation_aggregate}
\widetilde{\mathbf{z}}_t^{(b)}
=
\sum_{r \in N_t^{(b)}}
\alpha_{t,r}^{(b)}\mathbf{m}_r^{(b)},
\end{equation}
and the target representation and pointwise prediction take the following form
\begin{align}
\mathbf{z}_t^{(b)}
&=\tanh\!\left(\widetilde{\mathbf{z}}_t^{(b)}\right),
\label{eq:SeaGAT_target_rep}\\
\hat{y}_t^{(b)}
&=\mathrm{MLP}_{\mathrm{pred}}\!\left(\mathbf{z}_t^{(b)}\right)
\triangleq g_{\Theta}\!\left(\widetilde{\mathbf{z}}_t^{(b)}\right).
\label{eq:SeaGAT_prediction}
\end{align}
Because every reference is scored pointwise and \eqref{eq:preactivation_aggregate} uses summation, the mapping is thus invariant to the ordering of the selected references. 
A common translation of the coordinate frame applied to $\mathbf{p}_t$, $\mathbf{p}_b$, and all $\mathbf{p}_r$ also leaves the transmitter-centered coordinates and target--reference relations unchanged. 
These structural invariances follow directly from the relative coordinates and set aggregation.

\subsection{Analysis of Relation-Conditioned Aggregation}
\label{subsec:kernel_view}

Equation~\eqref{eq:preactivation_aggregate} is the first stage at which the node-message and relation-conditioned scoring pathways interact, i.e., messages enter linearly, while physical relations affect the aggregate only through the normalized weights. 
Analyzing this aggregate therefore isolates SeaGAT's evidence-allocation mechanism before the downstream nonlinearity and readout, whose local sensitivity can subsequently be revealed through the chain rule.
Throughout Sections~\ref{subsec:kernel_view} and~\ref{sec:arch}, we fix the designated transmitter $b$, set $O=O_b$ and $R^{\mathrm{obs}}=R_b^{\mathrm{obs}}$, and suppress the superscript from all transmitter-specific quantities. The target index $t$ is retained only where it distinguishes query-specific quantities. 
For the generic softmax allocation, differential identity, and single-score perturbation below, we additionally write $\mathcal{N}=N_t$, $s_r=s_{t,r}$, $\alpha_r=\alpha_{t,r}$, and $\widetilde{\mathbf{z}}=\widetilde{\mathbf{z}}_t$. 

Let $\boldsymbol{\alpha}$ and $\mathbf{s}$ collect $\{\alpha_r\}_{r\in\mathcal{N}}$ and $\{s_r\}_{r\in\mathcal{N}}$, respectively, define the probability simplex $\Delta_K=\{\mathbf{a}\succeq\mathbf{0}:\mathbf{1}^{\top}\mathbf{a}=1\}$, and let $H(\mathbf{a})=-\sum_r a_r\log a_r$ denote Shannon entropy. 
The softmax allocation in \eqref{eq:attention_weight} is equivalently characterized by the entropy-regularized problem \cite{2017_RegularizedAttention_Niculae,2020_FenchelYoung_Blondel}
\begin{equation}
\boldsymbol{\alpha}
=
\arg\max_{\mathbf{a}\in\Delta_K}
\left\{
\mathbf{a}^{\top}\mathbf{s}
+\tau H(\mathbf{a})
\right\}.
\label{eq:entropy_allocation}
\end{equation}
The linear term $\mathbf{a}^{\top}\mathbf{s}$ rewards allocation to high-scoring references, whereas $H(\mathbf{a})$ discourages complete concentration. 
The temperature $\tau$ controls the balance between score preference and allocation entropy. 
For finite scores and $\tau>0$, the objective is strictly concave, and its first-order condition gives $a_r\propto\exp(s_r/\tau)$; enforcing $\sum_r a_r=1$ recovers \eqref{eq:attention_weight} as the unique maximizer. 
Since \eqref{eq:preactivation_aggregate} is a convex combination of the node messages, SeaGAT can be interpreted as enabling an entropy-regularized barycentric allocation of reference evidence.

At the operator level, the same aggregate can be written as a learned normalized smoother in the latent message space. Following the normalized-smoothing form of Nadaraya--Watson regression \cite{1964_Kernel_Nadaraya,1964_Kernel_Geoffrey}, define the nonnegative relevance function
\begin{align}
\kappa_{\Theta}(t,r;O)
&=\mathbb{I}\!\left\{r\in N_t\right\}
\exp\!\left(s_{t,r}/\tau\right),\\
\bar{\kappa}_{\Theta}(t,r;O)
&=\frac{\kappa_{\Theta}(t,r;O)}
{\sum_{j\in R^{\mathrm{obs}}}\kappa_{\Theta}(t,j;O)}.
\label{eq:normalized_kernel}
\end{align}
Then
\begin{equation}
\widetilde{\mathbf{z}}_t
=
\sum_{r\in R^{\mathrm{obs}}}
\bar{\kappa}_{\Theta}(t,r;O)\mathbf{m}_r,
\label{eq:kernel_smoother_latent}
\end{equation}
which retains the normalized structure of classical kernel smoothing while operating on learned messages. Here, $\kappa_{\Theta}$ denotes a nonnegative relevance function in the Nadaraya--Watson sense. Within the sampled support, it depends on $\mathbf{p}_r-\mathbf{p}_b$, $y_r$, $d_{t,r}$, $\theta_{t,r}$, and optional $c_{t,r}$, making the resulting smoother transmitter-conditioned and measurement-dependent.

\begin{lemma}[]
\label{lem:differential_identity}
For arbitrary infinitesimal perturbations of the messages and scores, the aggregate differential satisfies
\begin{equation}
d\widetilde{\mathbf{z}}
=
\sum_{r\in\mathcal{N}}\alpha_r\,d\mathbf{m}_r
+
\frac{1}{\tau}
\sum_{r\in\mathcal{N}}
\alpha_r\left(\mathbf{m}_r-\widetilde{\mathbf{z}}\right)ds_r.
\label{eq:differential_identity}
\end{equation}
\end{lemma}
\begin{proof}
Differentiating the softmax gives
\begin{equation}
d\alpha_r
=
\frac{\alpha_r}{\tau}
\left(ds_r-\sum_{j\in\mathcal{N}}\alpha_jds_j\right).
\label{eq:softmax_differential}
\end{equation}
Applying the product rule to \eqref{eq:preactivation_aggregate} and substituting \eqref{eq:softmax_differential} yields
\begin{align}
d\widetilde{\mathbf{z}}
&=\sum_r\alpha_r\,d\mathbf{m}_r
+\sum_r\mathbf{m}_r\,d\alpha_r\notag\\
&=\sum_r\alpha_r\,d\mathbf{m}_r
+\frac{1}{\tau}
\left(
\sum_r\alpha_r\mathbf{m}_rds_r
-\widetilde{\mathbf{z}}\sum_r\alpha_rds_r
\right),
\end{align}
which rearranges to \eqref{eq:differential_identity}.
\end{proof}
Equation~\eqref{eq:differential_identity} decomposes aggregate variation into two pathways. 
The first term is the attention-weighted direct variation of the message contents under the current allocation. 
The second term captures score-induced weight redistribution, i.e., a perturbation of $s_r$ shifts the aggregate through the offset $\mathbf{m}_r-\widetilde{\mathbf{z}}$, which points from the current aggregate to reference message $r$. 
Adding the same constant to every score leaves the softmax weights unchanged. 
Equivalently, $\sum_{r\in\mathcal{N}}\alpha_r(\mathbf{m}_r-\widetilde{\mathbf{z}})=\mathbf{0}$, therefore a common score shift contributes no aggregate variation.
\begin{proposition}[]
\label{prop:relation_score_shift}
Fix one reference $r\in\mathcal{N}$. Suppose a change in its edge relation changes only $s_r$ to $s_r+\Delta$, while all messages and the other scores remain fixed. Let $q=\exp(\Delta/\tau)$ and denote the resulting aggregate by $\widetilde{\mathbf{z}}(\Delta)$, with $\widetilde{\mathbf{z}}(0)=\widetilde{\mathbf{z}}$. Then
\begin{equation}
\widetilde{\mathbf{z}}(\Delta)
-
\widetilde{\mathbf{z}}
=
\frac{\alpha_r(q-1)}
{1+\alpha_r(q-1)}
\left(\mathbf{m}_r-\widetilde{\mathbf{z}}\right).
\label{eq:relation_score_shift}
\end{equation}
\end{proposition}
\begin{proof}
Let $u_i=\exp(s_i/\tau)$ and $Z=\sum_i u_i$, so that $\alpha_i=u_i/Z$. After the perturbation, only $u_r$ changes, from $u_r$ to $q u_r$. The new normalizer is therefore
\begin{equation}
Z'
=
Z+(q-1)u_r
=
Z\left[1+\alpha_r(q-1)\right].
\label{eq:perturbed_denominator}
\end{equation}
Define $D_\Delta=1+\alpha_r(q-1)=(1-\alpha_r)+\alpha_r q>0$. The perturbed weights are then $\alpha_r'=q\alpha_r/D_\Delta$ and $\alpha_i'=\alpha_i/D_\Delta$ for $i\neq r$. Separating the $r$-th contribution from the remaining messages gives
\begin{align}
\widetilde{\mathbf{z}}(\Delta)
&=
\frac{q\alpha_r\mathbf{m}_r+
\sum_{i\neq r}\alpha_i\mathbf{m}_i}{D_\Delta}
\notag\\
&=
\frac{\widetilde{\mathbf{z}}+
\alpha_r(q-1)\mathbf{m}_r}{D_\Delta}.
\label{eq:perturbed_aggregate}
\end{align}
Since $D_\Delta-1=\alpha_r(q-1)$, subtracting $\widetilde{\mathbf{z}}$ from the last expression yields \eqref{eq:relation_score_shift}.
\end{proof}
For $\Delta>0$, the coefficient in \eqref{eq:relation_score_shift} lies in $(0,1)$, so the aggregate moves toward $\mathbf{m}_r$; for $\Delta<0$, it moves in the opposite direction. As $\Delta\rightarrow+\infty$, the aggregate approaches $\mathbf{m}_r$. When $K\geq2$, the limit $\Delta\rightarrow-\infty$ removes reference $r$ and renormalizes the remaining messages. The corresponding local response is
\begin{equation}
\lim_{\Delta\rightarrow0}
\frac{\widetilde{\mathbf{z}}(\Delta)-\widetilde{\mathbf{z}}}{\Delta}
=
\frac{\alpha_r}{\tau}
\left(\mathbf{m}_r-\widetilde{\mathbf{z}}\right),
\label{eq:relation_score_shift_local}
\end{equation}
which agrees with Lemma~\ref{lem:differential_identity}. Equation~\eqref{eq:relation_score_shift} is also an exact score-shift relation for finite $\Delta$. It therefore remains valid when bucketized edge encoding produces a discrete score jump, even when the score is nondifferentiable with respect to the underlying distance or bearing.

\begin{proposition}[]
\label{prop:measurement_influence}
Fix a selected reference $r\in N_t$. At any point where the
message and raw-score mappings are differentiable with respect
to $y_r$, the local sensitivity of the pre-activation aggregate
satisfies
\begin{equation}
\frac{\partial\widetilde{\mathbf{z}}_t}{\partial y_r}
=
\alpha_{t,r}
\frac{\partial\mathbf{m}_r}{\partial y_r}
+
\frac{\alpha_{t,r}}{\tau}
\left(\mathbf{m}_r-\widetilde{\mathbf{z}}_t\right)
\frac{\partial s_{t,r}}{\partial y_r}.
\label{eq:measurement_influence}
\end{equation}
\end{proposition}
\begin{proof}
Consider a point at which the node and score mappings are differentiable. The selected support $N_t$ remains fixed under an infinitesimal change in $y_r$, because the $k$-NN sampler used in this work depends on the reference coordinates rather than their RSS values. Moreover, the message and raw score are computed pointwise for each reference before softmax normalization. Hence, for every $i\in N_t$,
\begin{align}
d\mathbf{m}_i
&=
\mathbb{I}\{i=r\}
\frac{\partial\mathbf{m}_r}{\partial y_r}\,dy_r,
\notag\\
ds_{t,i}
&=
\mathbb{I}\{i=r\}
\frac{\partial s_{t,r}}{\partial y_r}\,dy_r.
\label{eq:rss_local_differentials}
\end{align}
Substituting \eqref{eq:rss_local_differentials} into Lemma~\ref{lem:differential_identity} leaves only the $r$-th term in each sum and gives
\begin{equation*}
d\widetilde{\mathbf{z}}_t
=
\left[
\alpha_{t,r}
\frac{\partial\mathbf{m}_r}{\partial y_r}
+
\frac{\alpha_{t,r}}{\tau}
\left(\mathbf{m}_r-\widetilde{\mathbf{z}}_t\right)
\frac{\partial s_{t,r}}{\partial y_r}
\right]dy_r.
\end{equation*}
Dividing by $dy_r$ yields \eqref{eq:measurement_influence}.
\end{proof}
The first term in \eqref{eq:measurement_influence} changes the message content contributed by reference $r$, while the second changes how strongly that reference is weighted through its relevance score. 
They reinforce when their inner product is positive and partially cancel when it is negative; exact cancellation requires equal magnitudes and opposite directions. 
Mechanistically, cancellation occurs when the RSS-induced change in the transferred message is opposed by the score-induced movement of the aggregate.
For example, a reference located across a blockage boundary may exhibit an RSS change that shifts its encoded message in one direction, while the learned scorer interprets the resulting node--relation combination as less relevant to the target and reduces its attention weight. 
The resulting redistribution can oppose the direct message change, producing partial cancellation.

\subsection{Bounded-Support Deviation and Complexity}
\label{sec:arch}

Lemma~\ref{lem:differential_identity} and Propositions~\ref{prop:relation_score_shift}--\ref{prop:measurement_influence} characterize message and score perturbations while the sampled support $N_t$ remains fixed. We now quantify the discrete aggregate change caused by replacing the full observation pool $R^{\mathrm{obs}}$ with the bounded support $N_t$. 
Since SeaGAT computes each raw score and message pointwise before normalization, evaluating the same learned functions on both supports leaves the retained scores and messages unchanged and isolates the effect of support truncation. Define
\begin{align}
\alpha_{t,r}^{\mathrm{full}}
&=\frac{\exp(s_{t,r}/\tau)}
{\sum_{j\in R^{\mathrm{obs}}}\exp(s_{t,j}/\tau)},\\
\widetilde{\mathbf{z}}_t^{\mathrm{full}}
&=\sum_{r\in R^{\mathrm{obs}}}
\alpha_{t,r}^{\mathrm{full}}\mathbf{m}_r,\\
\varepsilon_{t,K}
&=\sum_{r\in R^{\mathrm{obs}}\setminus N_t}\alpha_{t,r}^{\mathrm{full}}.
\label{eq:omitted_mass}
\end{align}
For $R^{\mathrm{obs}}\setminus N_t\neq\varnothing$, define the cross-support message diameter
\begin{equation}
D_{m,t,K}^{\mathrm{cross}}
=
\max_{\substack{i\in N_t\\j\in R^{\mathrm{obs}}\setminus N_t}}
\left\|\mathbf{m}_i-\mathbf{m}_j\right\|_2,
\label{eq:cross_support_diameter}
\end{equation}
Set $D_{m,t,K}^{\mathrm{cross}}=0$ when $R^{\mathrm{obs}}\setminus N_t=\varnothing$, and let $D_m=\max_{i,j\in R^{\mathrm{obs}}}\|\mathbf{m}_i-\mathbf{m}_j\|_2$ denote the full message-set diameter.

\begin{corollary}[]
\label{cor:bounded_support}
The actual $K$-reference aggregate in \eqref{eq:preactivation_aggregate} satisfies
\begin{equation}
\left\|
\widetilde{\mathbf{z}}_t^{\mathrm{full}}
-
\widetilde{\mathbf{z}}_t
\right\|_2
\leq
\varepsilon_{t,K}D_{m,t,K}^{\mathrm{cross}}
\leq
\varepsilon_{t,K}D_m.
\label{eq:bounded_support}
\end{equation}
\end{corollary}
\begin{proof}
If $\varepsilon_{t,K}=0$, the omitted references carry no full-support attention mass and $\widetilde{\mathbf{z}}_t^{\mathrm{full}}=\widetilde{\mathbf{z}}_t$, so \eqref{eq:bounded_support} holds. Now suppose $\varepsilon_{t,K}>0$. Since $N_t$ contains at least one reference and all finite-score softmax weights are positive, $0<\varepsilon_{t,K}<1$. The retained raw scores are identical under the full and bounded supports; therefore, for every $i\in N_t$,
\begin{equation*}
\alpha_{t,i}
=
\frac{\alpha_{t,i}^{\mathrm{full}}}
{1-\varepsilon_{t,K}}.
\end{equation*}
Define the normalized omitted-support weights $\beta_{t,j}=\alpha_{t,j}^{\mathrm{full}}/\varepsilon_{t,K}$ for $j\in R^{\mathrm{obs}}\setminus N_t$, and let
\begin{align}
\widetilde{\mathbf{z}}_{t,\mathrm{out}}
&=
\sum_{j\in R^{\mathrm{obs}}\setminus N_t}
\beta_{t,j}\mathbf{m}_j,
\notag\\
\widetilde{\mathbf{z}}_t^{\mathrm{full}}
&=
\left(1-\varepsilon_{t,K}\right)
\widetilde{\mathbf{z}}_t
+
\varepsilon_{t,K}\widetilde{\mathbf{z}}_{t,\mathrm{out}}.
\label{eq:full_support_decomposition}
\end{align}
Both $\{\alpha_{t,i}\}_{i\in N_t}$ and $\{\beta_{t,j}\}_{j\in R^{\mathrm{obs}}\setminus N_t}$ are nonnegative and sum to one. Consequently,
\begin{equation*}
\widetilde{\mathbf{z}}_{t,\mathrm{out}}
-
\widetilde{\mathbf{z}}_t
=
\sum_{\substack{i\in N_t\\
j\in R^{\mathrm{obs}}\setminus N_t}}
\alpha_{t,i}\beta_{t,j}
\left(\mathbf{m}_j-\mathbf{m}_i\right).
\end{equation*}
The coefficients in this sum are nonnegative and sum to one. The triangle inequality and the definition of $D_{m,t,K}^{\mathrm{cross}}$ therefore give $\|\widetilde{\mathbf{z}}_{t,\mathrm{out}}-\widetilde{\mathbf{z}}_t\|_2\leq D_{m,t,K}^{\mathrm{cross}}$. Subtracting $\widetilde{\mathbf{z}}_t$ from \eqref{eq:full_support_decomposition} then yields
\begin{equation*}
\left\|
\widetilde{\mathbf{z}}_t^{\mathrm{full}}
-
\widetilde{\mathbf{z}}_t
\right\|_2
=
\varepsilon_{t,K}
\left\|
\widetilde{\mathbf{z}}_{t,\mathrm{out}}
-
\widetilde{\mathbf{z}}_t
\right\|_2
\leq
\varepsilon_{t,K}D_{m,t,K}^{\mathrm{cross}}.
\end{equation*}
Finally, $D_{m,t,K}^{\mathrm{cross}}\leq D_m$ proves the second inequality in \eqref{eq:bounded_support}.
\end{proof}
Corollary~\ref{cor:bounded_support} identifies two conditions under which a bounded reference support remains close to the same-parameter full-support aggregate. First, $\varepsilon_{t,K}$ is small when the sampler retains the references that receive most of the full-support attention mass. This can occur when query-relevant propagation evidence is predominantly local and the $k$-NN support covers that neighborhood. Second, $D_{m,t,K}^{\mathrm{cross}}$ is small when the retained and omitted observations are mapped into a compact common region of latent message space, as may occur in a locally homogeneous propagation region where the omitted evidence is largely redundant. 
Although the aforementioned conditions characterize a favorable sampled support, we note that the sampler considered in this paper, e.g., $k$-NN, excludes the minimization of the omitted attention mass or the resulting deviation. 
Thus, the sampler optimization remains as an interesting piece of future work.

Let $M \triangleq |O|$ denote the size of the fixed transmitter-specific observation pool, let $K \triangleq |N_t|$ denote the number of selected references per query, and let $E$ denote the hidden representation width.
For one target--transmitter query, let $C_{\mathcal{S}}(M,K)$ denote the cost of the sampling operator in \eqref{ReferenceSet}.
The current $k$-NN implementation can achieve sublinear average retrieval cost in $M$ with spatial indexing structures such as $k$-d trees \cite{1975_KDTree}.
Node and edge encoding, attention scoring, and weighted aggregation on the sampled star graph require $O(KE^2)$ operations because the graph contains only $K$ target--reference edges.
The per-query inference cost is therefore
\begin{equation}
\label{eq:complexity}
O\!\left(
C_{\mathcal{S}}(M,K) + KE^2
\right).
\end{equation}
The overall training and testing costs scale linearly with the number of processed queries, and training adds a linear factor in the number of epochs.
For bounded $K$, the graph-side computation is linear in the sampled reference count and independent of $M$ after sampling.
The queries are mutually independent and can be processed in parallel for online target-wise inference.

This profile differs from those of representative learning-based baselines considered in Section~\ref{sec:results}.
The hierarchical design of 5GNN \cite{2022_5GNN} complements per-query local graphs with a global graph built over the query points, so that the global-stage message passing couples the queries and inference is naturally performed at the batch level over the entire query set.
Transformer-based estimators such as STORM \cite{2025_STORM} treat the participating reference observations as a joint token set, and the self-attention among the tokens incurs a per-query cost that grows quadratically with the number of involved references.
\section{Results and Performance Evaluation}\label{sec:results}
In this section, SeaGAT is evaluated against learning-based and classical baselines under the unified protocol. 
We first evaluate the contribution of individual design components through ablation studies; and then validate SeaGAT performance under different parametric configurations. 
Finally, the scenario-wise site-specific RME performance is introduced.

\subsection{Datasets and Preprocessing}
All schemes are evaluated under a single unified data-preparation
pipeline so that the comparison reflects modeling capability rather
than differences in data handling. Site-specific propagation data are
generated with the DeepMIMO ray-tracing tool~\cite{2019_DeepMIMO} under
three urban scenarios, namely Denver, New York, and Houston. Unless
otherwise stated, all transmitters operate at $3.5$ GHz; the Houston $28$ GHz variant is used only in the cross-frequency study. 
Each scenario contains three transmitters, and each transmitter is treated as one site; a transmitter and the user terminal at every target location are each equipped with a single antenna, and the RSS at a target is obtained by averaging the sub-carrier-level powers across the signal band. 
Constant $43$ dBm transmitting power is assumed for all transmitters. 
For each scenario we discretize the served area into $2\,\mathrm{m}\times 2\,\mathrm{m}$ grids, and the linear-scale RSS within a grid is averaged to a single value before being converted to dB.
 
\subsection{Observation Protocol and Data Splits}
\label{subsec:protocol}
For every transmitter, a fixed fraction $\eta$ of its grid points is
selected to form the transmitter-specific observation pool $O_b$; we refer to $\eta$ as the \emph{sampling ratio}. 
The selected points play a dual role: they constitute the shared reference pool $O_b$ from which kNN sampling draws evidence, and they also serve as the labeled queries seen during training. Concretely, when a selected point is used as a training query, it is removed from its own reference pool and predicted from the remaining selected points, so that no target is ever used as its own reference. All non-selected grid points are held out as evaluation queries and are never used for supervision; they are answered against the same observation pool $O_b$ at test time. This arrangement keeps the observation set identical across the training and evaluation stages and provides a common leakage-free reference-usage protocol for all learning-based baselines. Unless otherwise stated, the selected points are obtained by $k$-means sampling of the grid coordinates, which yields a spatially representative pool; the effect of replacing it with uniform random sampling is studied in the learned-attention and sampling-mismatch analysis. 
To capture split-induced variability, three independent selections are drawn with different random seeds, and every reported number is the mean $\pm$ standard deviation over these three seeds. We adopt a scenario-wise training protocol: one set of parameters is trained per scenario and shared across all transmitters of that scenario. 
The parameter studies are conducted on the Denver and New York scenarios; the per-transmitter comparison additionally includes Houston, and the cross-frequency study uses the Houston $28$ GHz domain. The root-mean-squared error (RMSE) in dB is used as the evaluation metric and, unless otherwise stated, is reported across all transmitters of a scenario.
 
\subsection{Training and Implementation Details}
\label{subsec:impl}
Both branches of SeaGAT operate on normalized inputs. The distance encoder uses the same range scale $d_0=1000$~m in all scenarios. With $B_d=256$, \eqref{eq:bucket_width} gives a first-bucket width of approximately $c/B_d\approx4$~m, concentrating resolution on nearby references while retaining an unbounded final bucket. 
The reference RSS is normalized with the Z-score method using statistics computed from the non-zero training-set RSS of each scenario. 
The bearing relation is encoded with $B_\theta=256$ circular buckets.
 
The hidden representation width is $E=128$. 
Each geometric relation is embedded into a $16$-dimensional vector, and the link-state agreement shares the same width; the resulting per-edge relation vector is mapped to the $E$-dimensional edge embedding by a single-hidden-layer perceptron of width $64$. 
The reference-node branch embeds the transmitter-centered coordinate and the normalized RSS, concatenates them, and fuses them through a single-hidden-layer perceptron of width $128$. 
The attention module is a single LeakyReLU-activated layer of width $64$ with negative slope $0.1$, applied to the concatenation of the node and edge embeddings, and the attention logits are normalized with a softmax temperature $\tau=0.7$. The transferred message is a linear map of the node embedding alone, and the readout that maps the aggregated representation of \eqref{eq:SeaGAT_target_rep} to the RSS estimate is a single-hidden-layer perceptron of width $E$. Unless stated otherwise, the default operating point uses $K=8$ references per query, sampling ratio $\eta=0.15$, and $k$-means sampling to form $O_b$; later experiments identify only the varied quantity or any deviation from this setting. Under this configuration, SeaGAT has approximately $120$k trainable parameters. 
It is worth noting that the resolution--model-size tradeoff is governed by $B_d$, e.g., with $16$-dimensional distance embeddings, $256$ and $128$ bins contribute $4096$ and $2048$ table parameters, respectively. The default setting therefore adds only $2048$ parameters over the $128$-bin variant, about $1.7\%$ of the complete model.
 
The model is trained to minimize a smooth-$L_1$ (Huber) regression objective. 
The Huber threshold is specified in the original RSS scale as $\delta_{\rm dB}=2$ dB and is converted to the
standardized scale as $\delta=\delta_{\rm dB}/\sigma_y$, where $\sigma_y$ is the training-set RSS standard deviation. 
 Optimization uses Adam with a learning rate of $5\times 10^{-4}$ and $(\beta_1,\beta_2)=(0.9,0.999)$ for $20$ epochs, with parameters updated on effective mini-batches of $128$ query samples and the training queries reshuffled every epoch. A held-out validation set of $20\%$ of the evaluation queries is scored every two epochs, and the checkpoint with the lowest validation RMSE is retained for testing. The learning-based baselines follow the same optimizer family and the same validation-based checkpoint selection, but their training budgets are not forced to match that of SeaGAT.
 This is because the baselines differ in architecture and convergence behavior; each is trained for as many epochs as its own validation curve requires to reach saturation, and several therefore use a substantially larger epoch count. 
For every learning-based method and random seed, each training-side configuration in the parameter studies is trained from scratch and undergoes independent validation-based checkpoint selection. 
\begin{table*}[!t]
\centering
\scriptsize
\caption{Ablation of SeaGAT components on the Denver and New York scenarios. RMSE is reported as mean $\pm$ std over three seeds; $\Delta$ denotes the RMSE change of each variant relative to the full model (positive means worse).}
\label{tab:SeaGAT_ablation}
%\resizebox{\textwidth}{!}{%
\begin{tabular}{lcccc}
\toprule
\textbf{Ablation setting} 
& \textbf{Denver} 
& $\Delta$
& \textbf{New York} 
& $\Delta$ \\
\midrule
\textbf{Full model} 
& $4.945 \pm 0.039$ 
& $0.000$ 
& $6.261 \pm 0.055$ 
& $0.000$ \\
\midrule
\multicolumn{5}{l}{\textit{Structure: aggregation and readout}} \\
\quad Uniform attention weights 
& $7.088 \pm 0.032$ 
& $+2.142$ 
& $7.839 \pm 0.063$ 
& $+1.578$ \\
\quad Linear readout (0 hidden layers) 
& $\mathbf{4.901 \pm 0.057}$ 
& $-0.044$ 
& $6.283 \pm 0.041$ 
& $+0.021$ \\
\quad Deeper readout (3 hidden layers) 
& $4.953 \pm 0.018$ 
& $+0.007$ 
& $6.296 \pm 0.077$ 
& $+0.035$ \\
\midrule
\multicolumn{5}{l}{\textit{Geometric relations on edges}} \\
\quad w/o target--reference relation 
& $5.979 \pm 0.068$ 
& $+1.034$ 
& $7.282 \pm 0.071$ 
& $+1.020$ \\
\quad All geometric relations 
& $5.086 \pm 0.057$ 
& $+0.140$ 
& $6.495 \pm 0.141$ 
& $+0.233$ \\
\quad + reference--transmitter relation 
& $5.095 \pm 0.051$ 
& $+0.149$ 
& $6.473 \pm 0.211$ 
& $+0.211$ \\
\quad + target--transmitter relation 
& $4.988 \pm 0.058$ 
& $+0.043$ 
& $6.345 \pm 0.127$ 
& $+0.084$ \\
\midrule
\multicolumn{5}{l}{\textit{Edge feature encoding}} \\
\quad Continuous distance 
& $5.523 \pm 0.085$ 
& $+0.577$ 
& $7.239 \pm 0.174$ 
& $+0.977$ \\
\quad Bucketized distance (128 bins) 
& $5.350 \pm 0.104$ 
& $+0.404$ 
& $6.547 \pm 0.132$ 
& $+0.285$ \\
\quad w/o link-state agreement 
& $5.630 \pm 0.059$ 
& $+0.685$ 
& $6.517 \pm 0.175$ 
& $+0.255$ \\
\quad Continuous angle 
& $5.069 \pm 0.022$ 
& $+0.123$ 
& $\mathbf{6.195 \pm 0.115}$ 
& $-0.067$ \\
\quad Bucketized angle (128 bins) 
& $4.976 \pm 0.076$ 
& $+0.031$ 
& $6.300 \pm 0.031$ 
& $+0.039$ \\
\bottomrule
\end{tabular}%
%}
\end{table*}

\subsection{Baseline Reproduction under the Unified Protocol}\label{subsec:baselines}
In practical deployments, building maps may be unavailable, outdated, or costly to acquire and maintain. 
We therefore adopt a building-map-free setting in which no method evaluated in this section is provided with building-map information as an input during training or inference.  
\subsubsection{5GNN}
We reproduce 5GNN~\cite{2022_5GNN} with its local--global GCN structure. Each query forms a local graph from the $K_{\rm local}$ nearest references with Gaussian radial-basis edge weights; the query RSS is masked, while reference nodes carry coordinates and observed RSS. Local chart embeddings are combined with a second GCN over a $K_{\rm global}$ query graph and mapped to RSS. 
For the common protocol, $K_{\rm local}=K$, with $K_{\rm local}=8$ at the default operating point; we report $K_{\rm global}\in\{5,10\}$ and train with Adam for $100$ epochs. The departures from the original study are the common reference budget, data split, and random seeds in Section~\ref{subsec:protocol}.
\subsubsection{CGFormer}
CGFormer uses cross-attention between a queried coordinate and sparse measurement tokens for grid-free prediction~\cite{CGFormer}. We retain its sampling-mask branch. The mask represents the complete observation-pool geometry, whereas the key/value context is restricted to the common $K$ nearest references. A training query is removed from both inputs to preserve leave-one-out prediction. 
The model uses mean-squared error and Adam with learning rate $5\times10^{-4}$, batch size $64$, and $200$ epochs.
\subsubsection{STORM}
STORM~\cite{2025_STORM} centers references on the target, rotates them by the measurement-weighted dominant direction, and augments the tokens with distance and angular features. We retain its two-head, four-block causal transformer with embedding width $48$ and prefix-wise supervision. The original square-patch context is replaced by the common $K$ nearest references, which preserves direct budget comparability and confines self-attention to the selected set. The model is trained with mean-squared error.
\subsubsection{RobUNet}
We reproduce the RobUNet encoder--decoder, residual double-convolution blocks, and channel/pixel attention~\cite{2024_RobUNet}. Its input is a single-channel sparse RSS map. Under the unified protocol, each supervised observation is removed from the input and predicted through pointwise masked self-supervision, preventing target leakage. Adam training is run for $100$ epochs, where validation performance saturates.

\subsubsection{Classical interpolation baselines}
We further compare against four standard non-learning interpolators, each applied independently per transmitter to the same observation set and evaluated at the same query locations as the learning-based methods. Ordinary kriging (OK) treats the observed RSS as a realization of a stationary random field with an exponential semivariogram,
$\gamma(h)=c_0+c\,[\,1-\exp(-h/a)\,]$, whose nugget $c_0$, partial sill $c$, and range $a$ are fit by weighted least squares to a $12$-bin empirical semivariogram; each query is then predicted through the ordinary-kriging system under the unbiasedness constraint. Universal kriging (UK) augments this model with a linear coordinate trend $[\,1,x,y\,]$, estimated by ridge-regularized least squares, and fits the same exponential semivariogram to the detrended residuals. The $K$-nearest-neighbor (KNN) estimator predicts each query as the unweighted mean of its eight nearest observations under the Euclidean metric. IDW forms a normalized weighted average over all valid observations of the designated transmitter,
\begin{equation}
\hat y_t^{\mathrm{IDW}}
=
\frac{\sum_{r\in R_b^{\mathrm{obs}}}w_{t,r}y_r^{(b)}}
{\sum_{r\in R_b^{\mathrm{obs}}}w_{t,r}},
\qquad
w_{t,r}=(d_{t,r}+10^{-12})^{-2}.
\end{equation}
If one or more references have zero distance to the query, their mean is returned as the zero-distance limit. None of these baselines involves iterative training. OK, UK, and IDW use the complete transmitter-specific observation pool for each prediction; KNN uses eight local observations; the pointwise learning baselines use $K=8$ measurement references at the default operating point, while RobUNet receives the complete sparse map under its masked supervision.
\subsection{Results and Performance Analysis}
\paragraph{Ablation of SeaGAT components}
As for the ablation, our aim is to validate the architecture design through three aspects illustrated in Table~\ref{tab:SeaGAT_ablation}. 
Each row changes only the listed component from the default SeaGAT configuration in Section~\ref{subsec:impl}; $\Delta$ denotes the corresponding RMSE change, with positive values indicating degradation.

\emph{Aggregation and readout structure.}
Uniform attention produces the largest degradation, with $\Delta=+2.142$ dB in Denver and $\Delta=+1.578$ dB in New York. 
It fixes the barycentric coefficients at $1/K$ and removes score-driven allocation. 
By comparison, linear and three-layer readouts remain within $0.044$ dB of the default head in Denver and $0.035$ dB in New York. 
The gain therefore arises primarily from learned score-driven evidence allocation before the readout, with limited sensitivity to post-aggregation depth.

\emph{Geometric relations on edges.}
Removing query--reference distance and bearing increases RMSE by $1.034$ dB in Denver and $1.020$ dB in New York. 
These relations change the score and therefore the aggregate-shift coefficient in \eqref{eq:relation_score_shift}. 
Adding reference--transmitter, target--transmitter, or both relation sets instead produces small degradations of $0.043$--$0.233$ dB. 
Reference--transmitter geometry is already present in $\mathbf{p}_{r}^{(b)}$, while target--transmitter geometry is common to all references of a query. The results support the compact query--reference edge relation.

\emph{Edge feature encoding.}
Distance encoding is the most sensitive relation-encoding choice. 
With the fixed scale $c$, the $256$- and $128$-bin settings provide near-range resolutions of approximately $4$ and $8$~m, respectively. 
A continuous distance input increases RMSE by $0.577/0.977$~dB in Denver/New York, while $128$-bin bucketization costs $0.404/0.285$~dB. 
As quantified in Section~\ref{subsec:impl}, the finer table adds only $2048$ parameters, so the default setting provides a favorable accuracy--model-size tradeoff. 
Removing link-state agreement costs $0.685/0.255$ dB. Bearing encoding is less sensitive: continuous and 128-bin variants change RMSE by at most $0.123$ dB.
\begin{figure}[!t]
	\centering
	\includegraphics[width=0.47\textwidth]{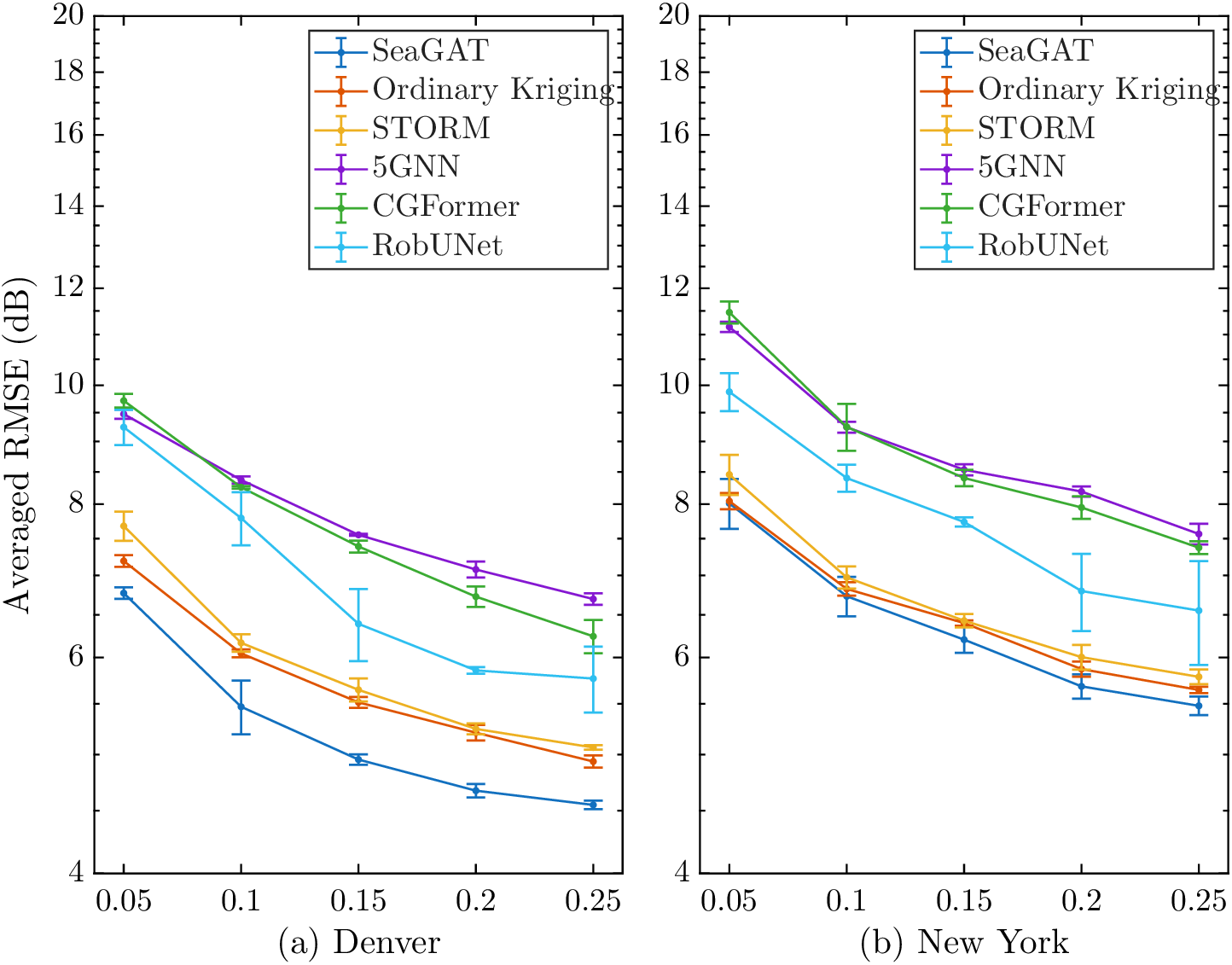}
	\caption{RMSE versus the sampling ratio $\eta$ in the Denver and New York scenarios. For pointwise learning-based methods, $K=8$ references are selected.}
	\label{RMSEvsSamplingRatio}
\end{figure}

\paragraph{RME performance evaluation under varying parameters}

\begin{figure}[!t]
	\centering
	\includegraphics[width=0.47\textwidth]{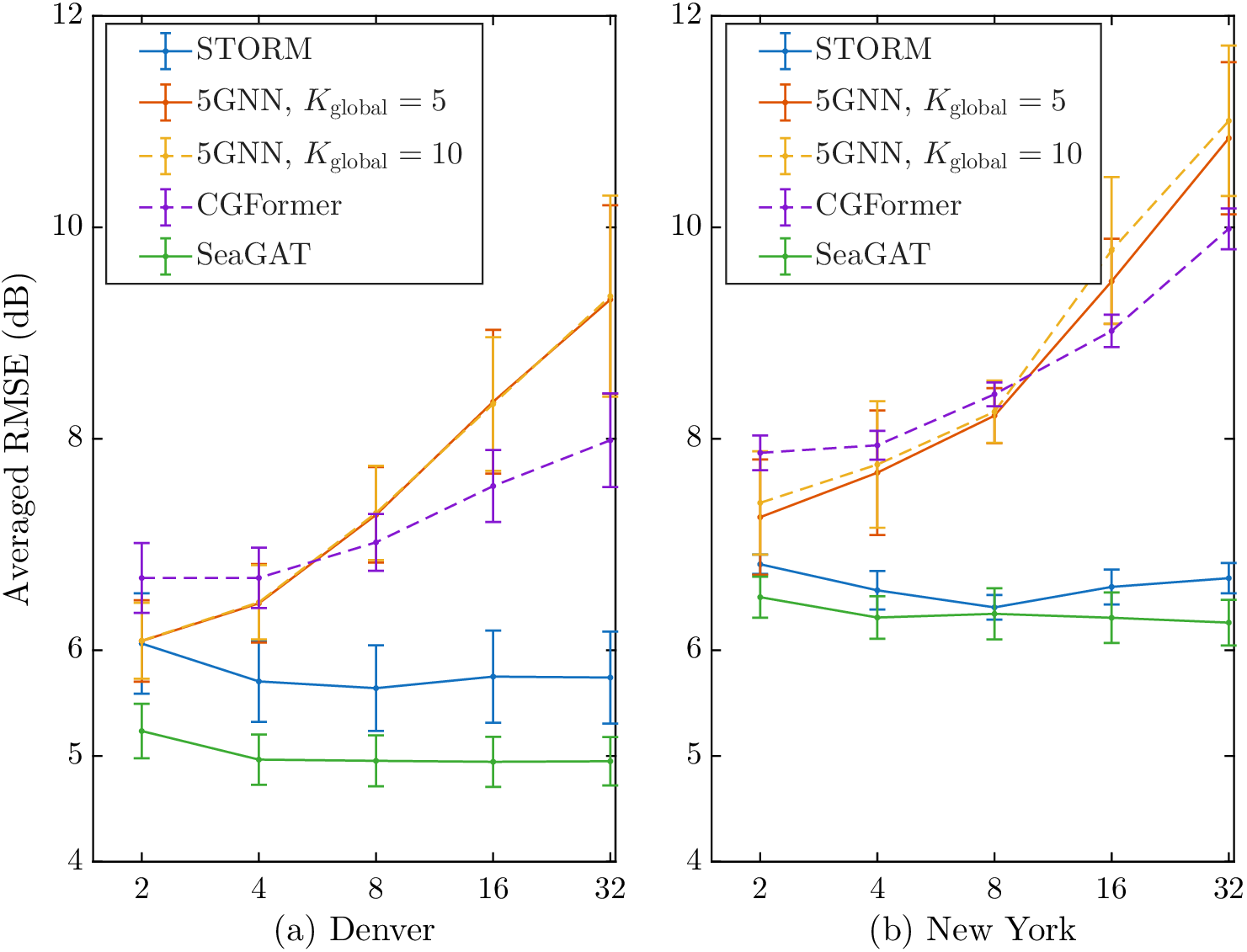}
	\caption{RMSE versus the number of kNN-selected references $K$ in the Denver and New York scenarios, at sampling ratio $\eta=0.15$.}
	\label{RMSEvsReferences}
\end{figure}
\figref{RMSEvsSamplingRatio} illustrates the RME performance under different sampling ratios. A larger $\eta$ simultaneously provides more labeled training queries and denser spatial measurement coverage, reducing the average distance from an evaluation target to the available observations at the expense of higher acquisition cost. 
The average RMSE therefore decreases with $\eta$ for every method in both scenarios, while the marginal improvement becomes smaller for most methods once $\eta\geq0.15$. 
SeaGAT achieves the lowest mean RMSE throughout the evaluated range; at $\eta=0.25$, it reaches approximately $4.5$ dB in Denver and $5.5$ dB in New York. 
CGFormer also benefits markedly from denser observations, yet remains above SeaGAT and STORM across the tested ratios. 
This ordering indicates that measurement density and the aggregation rule make distinct contributions to estimation accuracy.

\figref{RMSEvsReferences} examines the per-query reference budget $K$. SeaGAT reaches near-saturated accuracy by $K=4$ and remains the lowest-RMSE method in both scenarios; STORM also varies only mildly beyond this point. By contrast, the RMSE of CGFormer and 5GNN increases as larger reference sets are introduced, with the effect particularly visible in New York. Under the independently trained and validation-selected protocol in Section~\ref{subsec:impl}, this degradation reflects optimized performance with the larger context and cannot be attributed to a test-time reference-count mismatch. The result suggests that additional observations can introduce weakly relevant context that the aggregation mechanism does not fully suppress. SeaGAT's lower sensitivity is consistent with learned score-driven evidence allocation informed by explicit query--reference relations.

\paragraph{Learned attention behavior}
\begin{figure}[!t]
	\centering
	\includegraphics[width=\columnwidth]{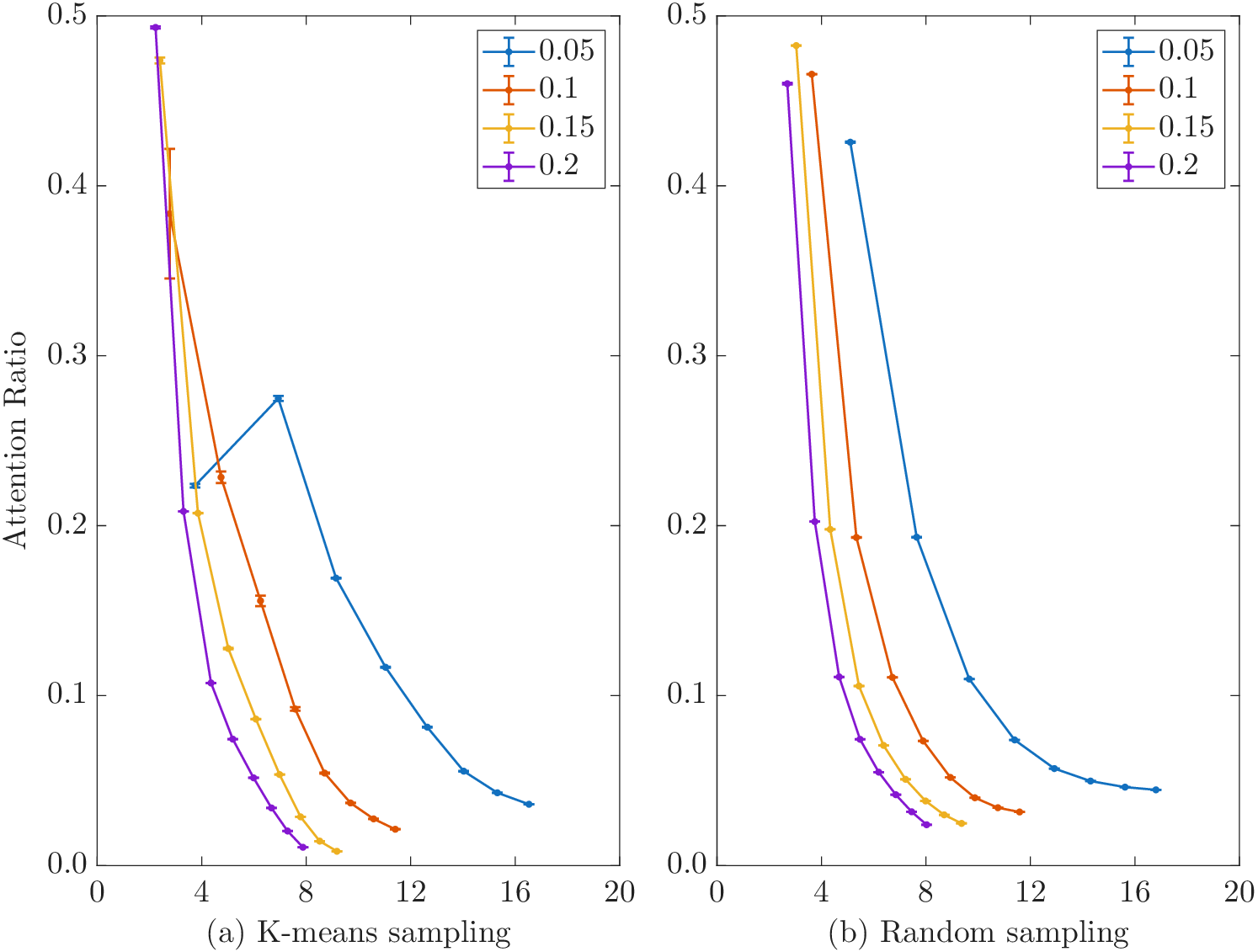}
	\caption{Attention weight versus the averaged target--reference distance (m), under $k$-means and random reference sampling, for several sampling ratios $\eta$ at $K=8$ in the Denver scenario.}
	\label{Attention_vs_AvgRefTgtDistance}
\end{figure}

\begin{figure}[!t]
	\centering
	\includegraphics[width=\columnwidth]{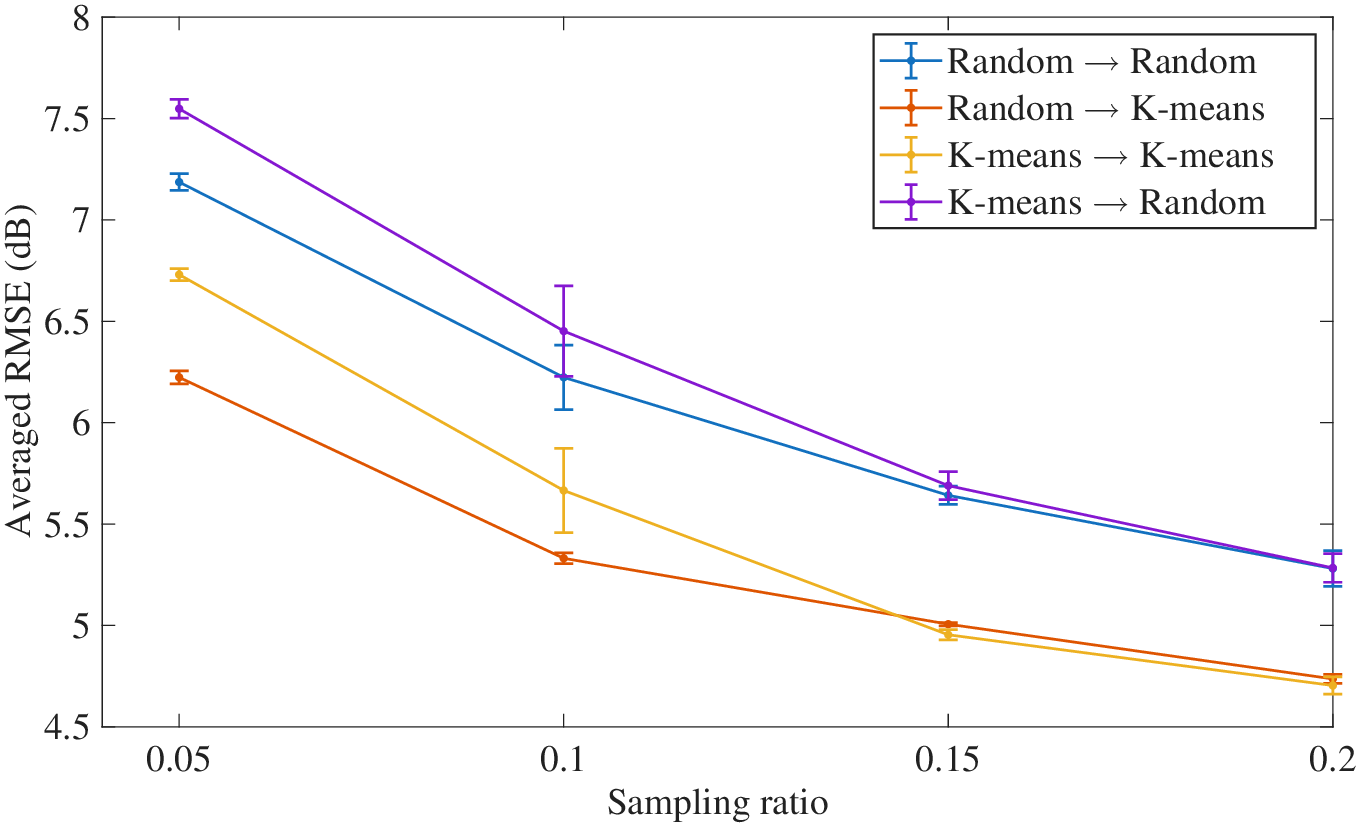}
	\caption{Effect of mismatched training/test reference sampling. RMSE
versus the sampling ratio $\eta$ for the four train$\rightarrow$test
combinations of random and $k$-means sampling, in the Denver scenario
at $K=8$.}
	\label{SamplingMethod_Train_vs_Test}
\end{figure}

\figref{Attention_vs_AvgRefTgtDistance} illustrates the calculated attention weight assigned to each reference, as a function of the averaged distance between each reference and the target. Each curve corresponds to a particular sampling ratio and is obtained by averaging over different targets. 

It can be observed that the attention weight assigned to the references generally follows the trend of \textit{the closer the distance, the higher the learned attention weight}. 
However, when $\eta=0.05$ with $k$-means sampling, the learned attention assigns a higher weight to the second nearest reference than to the nearest reference. 
When $\eta$ increases to $0.1$, the mean attention weight assigned to the nearest reference increases to $0.38$ but the seed-wise variance is still noticeable. 
As for the case with random sampling, consistent attention behavior is observed across all cases even though the averaged distance between the target and each reference increases with the reduction of $\eta$. 

The sampling-mismatch experiment in \figref{SamplingMethod_Train_vs_Test} compares the four train$\rightarrow$test combinations of random and $k$-means sampling.
The training sampler has negligible impact when $\eta\geq 0.15$, while a visible gap appears at lower sampling ratios.
At low reference ratios, this gap coincides with the less regular attention profile learned from $k$-means samples, including repeated cases in which the nearest reference receives a lower weight than the second nearest reference.
Across all train$\rightarrow$test combinations, $k$-means sampling at inference improves RMSE by approximately $0.5$ dB, indicating that its spatially representative reference set provides stronger evidence under the same budget.

\paragraph{Cross-scenario generalization capability}
In practical measurement-driven RME, retraining a model to fit different data distributions can be costly. 
A desirable estimator should therefore learn reference-to-target propagation relations that remain useful when the deployment scenario differs from the training scenario. 
We evaluate the cross-scenario generalization capability of each method by transferring the trained model parameters from source domains with different urban layouts and carrier frequencies to a common Houston 28 GHz target domain.
Specifically, each method is trained separately on Denver 3.5 GHz, New York 3.5 GHz, and Houston 3.5 GHz, and is then evaluated on Houston 28 GHz. The model trained directly on Houston 28 GHz is also reported as an in-domain reference.

As illustrated in \figref{crossscenario}, SeaGAT exhibits the most consistent cross-domain behavior among the compared learning-based methods.
For all source training domains and all tested values of $K$, SeaGAT achieves the lowest mean RMSE on the Houston 28 GHz target domain. 
More importantly, the SeaGAT curves obtained from different source domains almost overlap with each other and remain close to the in-domain Houston 28 GHz reference. 
Among the other attention-based query estimators, STORM and CGFormer exhibit different transfer behaviors. STORM maintains relatively stable performance across source domains, but its RMSE remains consistently higher than SeaGAT, indicating that its invariant coordinate representation improves robustness while providing limited adaptation to domain-specific propagation relations. 
CGFormer benefits from cross-attention based reference aggregation, yet its larger variation across source domains and reference budgets indicates that its learned query-reference interaction is more dependent on source-domain representation learning. 
In contrast, 5GNN is highly sensitive to both the source domain and the number of selected references.
As discussed earlier, under distribution shift, introducing more neighboring references may inject less transferable or even misleading messages into the graph update, which weakens the contribution of the most relevant local evidence.

\begin{figure}[!t]
	\centering
	\includegraphics[width=\columnwidth]{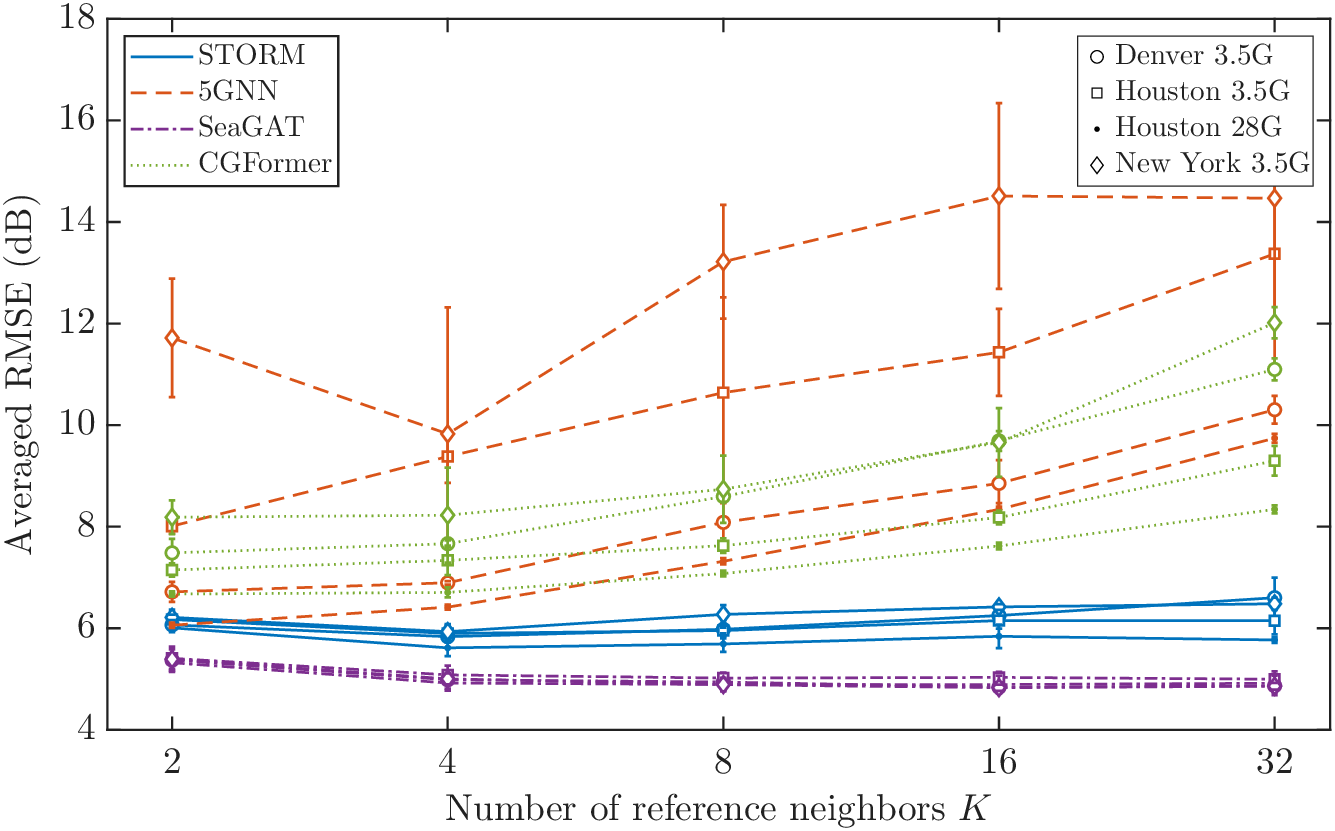}
	\caption{Cross-domain generalization on the Houston $28$ GHz target
domain at sampling ratio $\eta=0.15$. Line styles denote methods and markers denote the source training domain; the Houston $28$ GHz source is the in-domain reference.}
	\label{crossscenario}
\end{figure}

\paragraph{Site-wise performance evaluation}
Within the same scenario, different transmitters may exhibit markedly different propagation patterns because of local geometry and blockage conditions. Table~\ref{table:overall} reports the corresponding transmitter-wise RMSE. SeaGAT achieves the lowest mean RMSE on $7$ of the $9$ transmitters and leads all three transmitters in Denver and Houston. At the first two New York transmitters, OK provides slightly lower RMSE than SeaGAT. 
Among the learning-based baselines, STORM is the closest comparator across the nine transmitters, while CGFormer remains higher in RMSE under the same bounded measurement-reference protocol. IDW records the largest RMSE on all nine transmitters. Its fixed isotropic distance-decay law uses the complete observation pool and cannot adapt the relevance rule to blockage boundaries or directional propagation structure. Overall, SeaGAT provides the strongest performance across the evaluated transmitters, with localized complementary behavior from covariance-based kriging.

The representative maps in Figs.~\ref{Denver_site1} and~\ref{Houston_site1} support the same overall pattern and expose two recurrent error modes: local artifacts and excessive smoothing. In the Denver example, SeaGAT follows the narrow high-power corridors and preserves abrupt transitions around buildings more closely than the competing methods. 5GNN and CGFormer exhibit local intensity fluctuations, while KNN and IDW produce more pronounced patch-like artifacts. STORM, RobUNet, OK, and UK generate smoother fields but blur part of the transition structure. Some weak reflections remain unresolved by SeaGAT, indicating that the selected references and current relation attributes provide limited evidence for these paths.

The Houston example contains a more spatially continuous strong-signal region. SeaGAT again retains its dominant geometry and the blockage-induced boundaries with comparatively limited artifacts. The competing methods display the same two tendencies observed in Denver: local deviations for the measurement-aggregation estimators and boundary diffusion for the smoother map or covariance-based reconstructions. SeaGAT shows localized deviations near several building edges, leaving scope for broader sampling or richer propagation relations.

\begin{table*}[!t]
\centering
\caption{Per-transmitter RMSE (mean $\pm$ std, decibel) across scenarios at sampling ratio $\eta=0.15$ and $K=8$. The best result for each transmitter is in bold.}
\label{table:overall}
\scriptsize
\setlength{\tabcolsep}{3.5pt}
\renewcommand{\arraystretch}{1.15}
\resizebox{\textwidth}{!}{
\begin{tabular}{c c c c c c c c c c c}
\toprule
Scenario & Tx. & SeaGAT & 5GNN & CGFormer & STORM & RobUNet & OK & UK & kNN & IDW \\
\midrule

\multirow{3}{*}{\rotatebox[origin=c]{90}{\textbf{Denver}}}
& 1 & \textbf{5.207 $\pm$ 0.080} & 7.584 $\pm$ 0.113 & 7.174 $\pm$ 0.068 & 5.851 $\pm$ 0.063 & 6.329 $\pm$ 0.244 & 5.881 $\pm$ 0.137 & 5.571 $\pm$ 0.080 & 7.545 $\pm$ 0.053 & 8.565 $\pm$ 0.020 \\
& 2 & \textbf{4.686 $\pm$ 0.069} & 6.842 $\pm$ 0.168 & 6.694 $\pm$ 0.133 & 5.173 $\pm$ 0.044 & 7.071 $\pm$ 1.496 & 4.931 $\pm$ 0.072 & 4.931 $\pm$ 0.072 & 6.407 $\pm$ 0.043 & 8.301 $\pm$ 0.158 \\
& 3 & \textbf{5.023 $\pm$ 0.106} & 7.506 $\pm$ 0.119 & 7.193 $\pm$ 0.163 & 5.989 $\pm$ 0.107 & 6.287 $\pm$ 0.482 & 5.733 $\pm$ 0.117 & 5.733 $\pm$ 0.119 & 7.436 $\pm$ 0.054 & 8.656 $\pm$ 0.023 \\
\midrule

\multirow{3}{*}{\rotatebox[origin=c]{90}{\textbf{NY}}}
& 1 & 6.418 $\pm$ 0.136 & 7.973 $\pm$ 0.293 & 8.377 $\pm$ 0.117 & 6.523 $\pm$ 0.188 & 7.185 $\pm$ 0.156 & \textbf{6.371 $\pm$ 0.032} & 6.372 $\pm$ 0.034 & 7.705 $\pm$ 0.038 & 8.851 $\pm$ 0.147 \\
& 2 & 6.479 $\pm$ 0.044 & 8.366 $\pm$ 0.419 & 8.463 $\pm$ 0.164 & 6.507 $\pm$ 0.121 & 7.971 $\pm$ 1.288 & \textbf{6.334 $\pm$ 0.083} & 6.336 $\pm$ 0.078 & 8.040 $\pm$ 0.095 & 9.247 $\pm$ 0.115 \\
& 3 & \textbf{6.027 $\pm$ 0.064} & 8.490 $\pm$ 0.193 & 8.423 $\pm$ 0.067 & 6.565 $\pm$ 0.079 & 7.296 $\pm$ 0.357 & 6.496 $\pm$ 0.085 & 6.501 $\pm$ 0.084 & 7.988 $\pm$ 0.072 & 9.492 $\pm$ 0.061 \\
\midrule

\multirow{3}{*}{\rotatebox[origin=c]{90}{\textbf{Houston}}}
& 1 & \textbf{4.361 $\pm$ 0.172} & 7.063 $\pm$ 0.065 & 6.722 $\pm$ 0.200 & 5.700 $\pm$ 0.307 & 6.413 $\pm$ 1.198 & 5.458 $\pm$ 0.142 & 5.422 $\pm$ 0.133 & 7.154 $\pm$ 0.126 & 8.176 $\pm$ 0.143 \\
& 2 & \textbf{5.087 $\pm$ 0.192} & 7.000 $\pm$ 0.130 & 7.244 $\pm$ 0.105 & 5.794 $\pm$ 0.136 & 6.296 $\pm$ 0.199 & 7.012 $\pm$ 0.389 & 5.367 $\pm$ 0.177 & 6.958 $\pm$ 0.080 & 8.269 $\pm$ 0.127 \\
& 3 & \textbf{5.513 $\pm$ 0.015} & 7.745 $\pm$ 0.276 & 7.483 $\pm$ 0.065 & 6.068 $\pm$ 0.140 & 6.751 $\pm$ 0.349 & 5.635 $\pm$ 0.109 & 5.635 $\pm$ 0.109 & 7.542 $\pm$ 0.075 & 9.029 $\pm$ 0.127 \\

\bottomrule
\end{tabular}
}
\end{table*}

\begin{figure*}[!t]
	\centering
	\includegraphics[width=\textwidth]{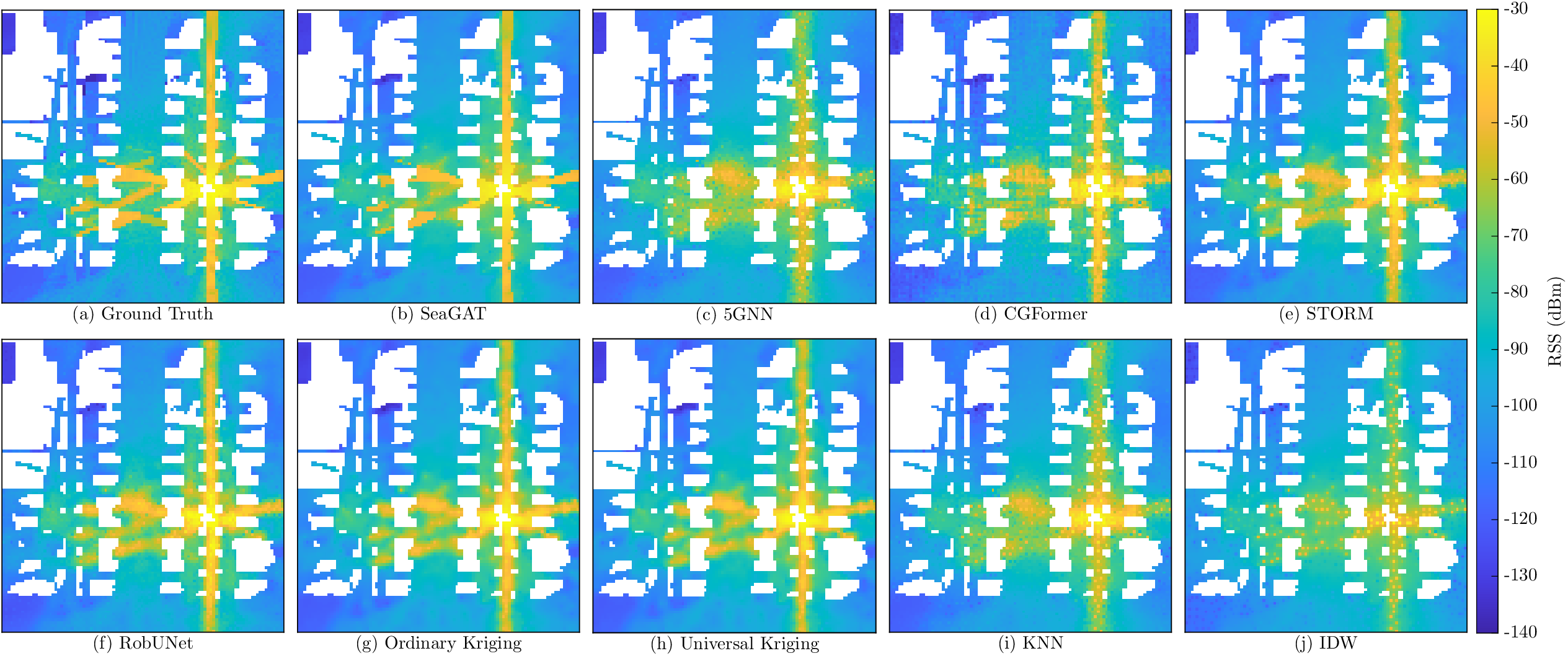}
	\caption{Qualitative comparison on a representative Denver transmitter. Estimated radio maps from SeaGAT and the baselines are shown against the ground truth.}
	\label{Denver_site1}
\end{figure*}
\begin{figure*}[!t]
	\centering
	\includegraphics[width=\textwidth]{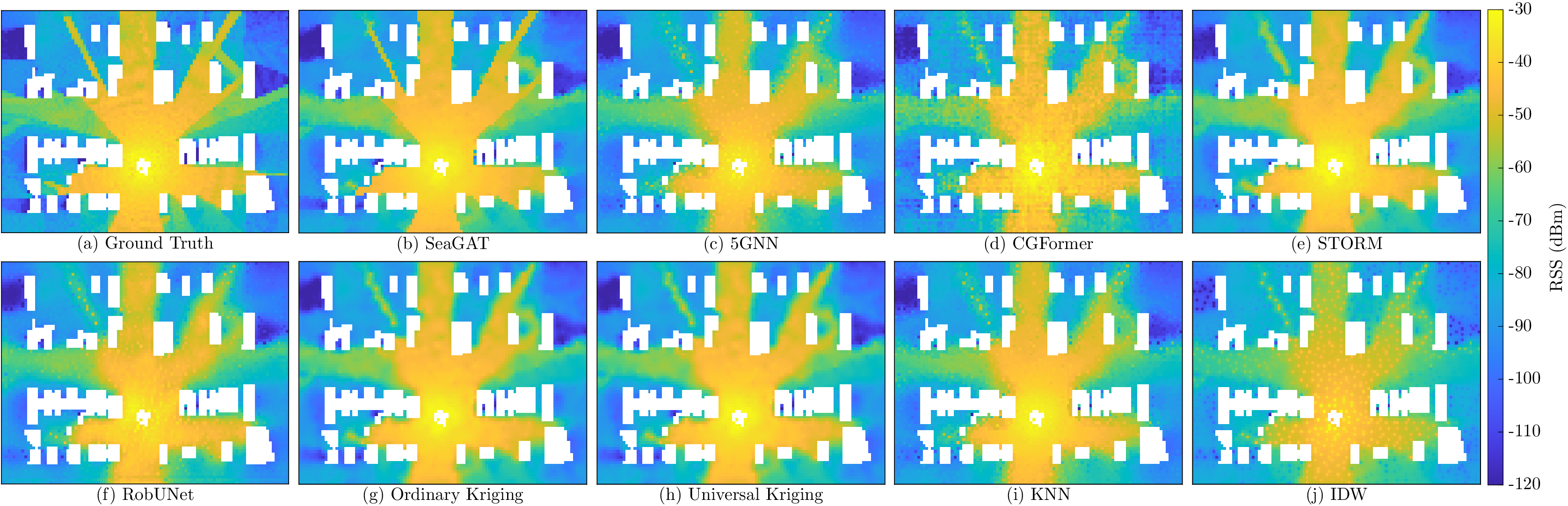}
	\caption{Qualitative comparison on a representative Houston transmitter. Estimated radio maps from SeaGAT and the baselines are shown against the ground truth. }
	\label{Houston_site1}
\end{figure*}
\section{Conclusion}\label{sec:conclusion}

This paper proposed SeaGAT for transmitter-resolved pointwise radio map estimation from sparse measurements. 
SeaGAT uses a bounded target-centered star graph to transfer node-derived measurement evidence through attention weights conditioned on explicit query--reference relations. Within a fixed support, the analysis shows that relation perturbations redistribute existing evidence, whereas reference-RSS perturbations affect both message content and relevance scoring; the bounded graph keeps graph-side computation linear in $K$, with support-truncation deviation controlled by omitted attention mass and cross-support message separation. At $\eta=0.15$ and $K=8$, SeaGAT achieves the lowest transmitter-wise RMSE in $7$ of $9$ cases. Replacing learned attention with uniform weights increases RMSE by $1.58$--$2.14$~dB, while removing query--reference relations increases it by $1.02$--$1.03$~dB. SeaGAT also achieves the lowest mean RMSE for every tested $K$ and every source training domain in transfer to Houston $28$~GHz. These results support structured query-conditioned evidence aggregation as an effective and transferable design for transmitter-resolved pointwise RME.

\bibliographystyle{IEEEtran}
\bibliography{RadioMap}

\end{document}